\documentclass[a4paper,UKenglish,cleveref, autoref]{lipics-v2019}

\hideLIPIcs
\title{A Single Exponential-Time FPT Algorithm for Cactus Contraction}

\titlerunning{A Single Exponential-Time FPT Algorithm for Cactus Contraction}

\nolinenumbers

\author{R. Krithika}{Indian Institute of Technology Palakkad, Palakkad, India}{krithika@iitpkd.ac.in}{}{}
\author{Pranabendu Misra}{Chennai Mathematical Institute, Chennai, India.}{pranabendu@cmi.ac.in}{}{}
\author{Prafullkumar Tale}{Indian Institute of Science
Education and Research, Pune, India}{prafullkumar@iiserpune.ac.in}{}{}

\authorrunning{R. Krithika, P. Misra, and P. Tale}


\relatedversion{An extended abstract of this article
appeared in \href{https://link.springer.com/chapter/10.1007/978-3-319-94776-1_29}{COCOON 2018} and full version appeared in \href{https://www.sciencedirect.com/science/article/pii/S0304397523001160?via\%3Dihub}{TCS 2023}.}

\acknowledgements{The first author is supported by SERB MATRICS grant number MTR/2022/000306.}


\usepackage{amsmath}
\usepackage{amsfonts}
\usepackage{amsthm}
\usepackage{amssymb}
\usepackage[ruled,vlined,algosection,linesnumbered]{algorithm2e}%algopart %algosection %algochapter
\usepackage{complexity} % For symbols in complexity class.
\usepackage{enumitem} % Control layout of itemize, enumerate, description
\newcommand{\yesinst}{\textsc{Yes}}

\newcommand{\yesalgo}{\texttt{Yes}}
\newcommand{\noalgo}{\texttt{No}}

\newcommand{\calF}{\mathcal{F}}

\newcommand{\calH}{{\mathcal H}}

\newcommand{\calO}{\ensuremath{{\mathcal O}}}
\newcommand{\OO}{\mathcal{O}}

\newcommand{\calW}{\mathcal{W}}
\newcommand{\calX}{\mathcal{X}}

\newtheorem{branching rule}{Branching Rule}%[section]
\newtheorem*{lemma*}{Lemma}
\newtheorem{marking-scheme}{Marking Scheme}[section]
\newtheorem{observation}{Observation}[section]
\newtheorem*{observation*}{Observation}

\newtheorem{reduction rule}{Reduction Rule}[section]
\newtheorem{Reduction Rule}{Reduction Rule}
\newtheorem{recoloring}{Recoloring}[section]
\newtheorem*{theorem*}{Theorem}

\newcommand{\defparproblem}[4]{
  \vspace{1mm}
\noindent\fbox{
  \begin{minipage}{0.96\textwidth}
  \begin{tabular*}{\textwidth}{@{\extracolsep{\fill}}lr} #1  & {\bf{Parameter:}} #3
\\ \end{tabular*}
  {\bf{Input:}} #2  \\
  {\bf{Question:}} #4
  \end{minipage}
  }
  \vspace{1mm}
}

\begin{document}

\maketitle
%\linenumbers


\begin{abstract}
For a collection $\mathcal{F}$ of graphs, 
the $\mathcal{F}$-\textsc{Contraction} problem takes 
a graph $G$ and an integer $k$ as input and decides if 
$G$ can be modified to some graph in $\mathcal{F}$
using at most $k$ edge contractions.
The $\mathcal{F}$-\textsc{Contraction} problem is 
\NP-Complete for several graph classes $\mathcal{F}$.
Heggerners et al.  [Algorithmica, 2014] initiated the study of 
$\mathcal{F}$-\textsc{Contraction} in the realm of parameterized complexity.
They showed that it is \FPT\ if $\mathcal{F}$ is the set of all trees or the set of all paths.
In this paper, we study $\mathcal{F}$-\textsc{Contraction} where $\mathcal{F}$ is 
the set of all cactus graphs and show that we can solve it in 
$2^{\calO(k)} \cdot |V(G)|^{\OO(1)}$ time.
\end{abstract}

\section{Introduction} 
For a collection $\mathcal{F}$ of graphs, 
the $\mathcal{F}$-\textsc{Modification} problem takes a 
graph $G$ with an integer $k$ as input and decides 
if $G$ can be modified to some graph in $\mathcal{F}$ using at most $k$ modifications.
The $\mathcal{F}$-\textsc{Modification} problem is an abstraction of well-studied problems like 
\textsc{Vertex Cover}, 
\textsc{Feedback Vertex Set}, 
\textsc{Odd Cycle Transversal}, and
\textsc{Minimum Fill-In}, to name a few.
The {\sc ${\cal F}$-Contraction} problems are those where edge contractions are the only modification operation allowed.
The {\em contraction} of the edge $e=uv$ in $G$ deletes vertices $u$ and $v$ from $G$, 
and adds a new vertex adjacent to vertices adjacent to $u$ or $v$. 
This process does not introduce self-loops or parallel edges, and the resulting graph 
is denoted by $G/e$. 
Early papers by Watanabe et al.~\cite{watanabe81, watanabe83} and Asano and Hirata~\cite{asano83} showed that {\sc ${\cal F}$-Contraction} is \NP-Complete even when $\calF$ is a collection of for several simple and well-structured graph classes such as paths, stars, trees, etc. 

Heggernes et al.~\cite{tree-contraction} initiated the study of these problems in the realm of parameterized complexity
with the solution size, i.e., the number of allowed contractions, $k$ as the parameter.
They presented single exponential time \FPT\ algorithms for contracting a given graph into a path and a tree.
{\sc ${\cal F}$-Contraction} problems have received much attention in parameterized complexity since then.
Graph contraction problems are more challenging than their vertex/edge deletion/addition counterparts.
One of the intuitive reasons is that the classical \emph{branching technique} does not work even for graph classes $\mathcal{F}$ that have a finite forbidden structure characterization.
In the case of vertex/edge deletion/addition operations, to destroy a structure that forbids the input graph from being in $\mathcal{F}$, one needs to include at least one vertex, edge, or non-edge from that structure in the solution.
This is not true in the case of contractions, as a contraction of an edge outside the forbidden structure can destroy it. 

Despite the inherent difficulties of contraction problems, we know several fixed-parameter tractability results when the parameter is $k$.
A series of papers presented \FPT\ algorithms for {\sc ${\cal F}$-Contraction} when $\calF$ is the collection of 
generalization and restrictions of trees~\cite{agarwal2019parameterized, agrawal2017paths}, 
bipartite graphs~\cite{guillemot2013faster,heggernes2013obtaining}, 
planar graphs~\cite{golovach2013obtaining},
grids~\cite{saurabh2020parameterized},
cliques~\cite{cai2013contracting}, 
bi-cliques~\cite{martin2015computational},
degree constrained graph classes~\cite{Belmonte:2014,golovach2013increasing,DBLP:journals/algorithmica/SaurabhT22}, etc.
Further, fixed-parameter intractability results are known for 
split graphs~\cite{agrawal2019split} and 
chordal graphs \cite{cai2013contracting, lokshtanov2013hardness}.

A \emph{cactus} is a graph in which every edge is in at most one cycle.
Hence, cactus graphs are a generalization of trees.
We formally define the problem addressed in this article.

\defparproblem{\textsc{Cactus Contraction}}{A graph $G$ and an integer $k$}{$k$}{Can we contract at most $k$ edges in $G$ to obtain a cactus?}

%Here, $G/F$ denotes the graph obtained from $G$ by contracting all edges in $F$ .
It is easy to verify that \textsc{Cactus Contraction} is in \NP, and its \NP-Completeness follows from \cite{lossy-fst}.
An edge contraction can reduce the \emph{treewidth}\footnote{See Chapter~$7$ in \cite{cygan2015parameterized} for the definition and application of the parameter.} of a graph by at most one.
Moreover, a cactus graph has treewidth at most $2$.
Hence, if a graph is $k$-contractible\footnote{Please refer to Section~\ref{sec:prelims} for formal definitions and terminologies.} to a cactus, then its treewidth is at most $k + 2$.
The standard application of dynamic programming over an optimum tree decomposition results in an algorithm running in time $2^{\calO(k \log k)} \cdot |V(G)|^{\calO(1)}$.
In this article,  we present an improved algorithm that runs in $2^{\calO(k)} \cdot |V(G)|^{\calO(1)}$ time.

\begin{figure}[t]
\centering
\includegraphics[scale=0.65]{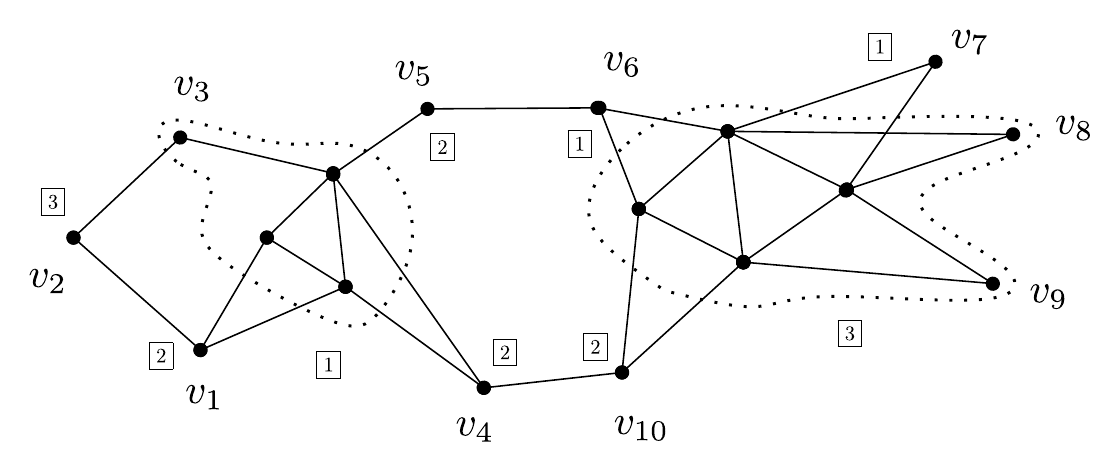}
\caption{A compatible coloring of an input graph.  
One of the non-trivial properties of such a coloring is that any color class (denoted by dotted lines) 
contains at most one big witness set.
\label{fig:fpt-cactus-1}}
\end{figure}

\paragraph*{Outline of the algorithm}
Our algorithm builds upon ideas presented in \cite{tree-contraction} but requires a more complex structural analysis of the input graph.
%We first introduce the essential terminology for presenting the outline of our algorithm.
We first argue that $(G, k)$ is a \yesinst\ instance of \textsc{Cactus Contraction}
if and only if
%We prove that determining whether an instance $(G, k)$ of \textsc{Cactus Contraction}
%is a \yesinst\ instance is equivalent to determining whether 
there is a partition of 
$V(G)$ (called a \emph{$T$-witness structure} of $G$ for some cactus $T$) 
that satisfies the following properties.
\begin{itemize}
\item Every part, called a \emph{witness set},  induces a connected subgraph of $G$.
Moreover, the number of edges in a spanning forest of these connected components is at most $k$.
\item Contracting all edges in a spanning forest of witness sets results in a cactus
(in particular, cactus $T$).
\end{itemize}
Given this equivalence, the task is to identify all \emph{big witness sets} (witness sets with at least two vertices) in a cactus witness structure of $G$.

As in \cite{tree-contraction},  we first present a randomized algorithm for \textsc{Cactus Contraction} on $2$-connected graphs, which consists of the following three steps.
%the algorithm for \textsc{Cactus Contraction} on general graphs consists of a randomized algorithm for \textsc{Cactus Contraction} on $2$-connected graphs as a subroutine and a derandomization procedure using universal sets.
%The randomized algorithm for \textsc{Cactus Contraction} on $2$-connected graphs takes as input a $2$-connected graph $G$ and an integer $k$ and returns \yesalgo\ or \noalgo. 
%The subroutine consists of the following three steps.
\begin{enumerate}
\item \textit{The algorithm constructs a coloring $f: V(G) \rightarrow \{1, 2, 3\}$ by 
assigning one of the colors to each vertex independently and uniformly at random.}
\end{enumerate}
We argue that if we can contract at most $k$ edges in $G$ to convert it into a cactus, 
then random coloring $f$ is `compatible' (see Definition~\ref{def:compatible})
with a fixed (but hypothetical) $T$-witness structure (see Definition~\ref{def:graph-contractioon})
of $G$ with a high probability
for some cactus $T$.
We say a connected component of $G[f^{-1}(i)]$ in \emph{non-trivial} if it contains at least two vertices.
Such a compatible coloring ensures that any non-trivial connected component of $G[f^{-1}(i)]$ contains 
at most one big witness set for every $i \in \{1, 2, 3\}$.
Figure~\ref{fig:fpt-cactus-1} shows a (compatible) coloring of a given graph.
We formalize this step in Section~\ref{sec:coloring}.

\begin{figure}[t]
\centering
\includegraphics[scale=0.65]{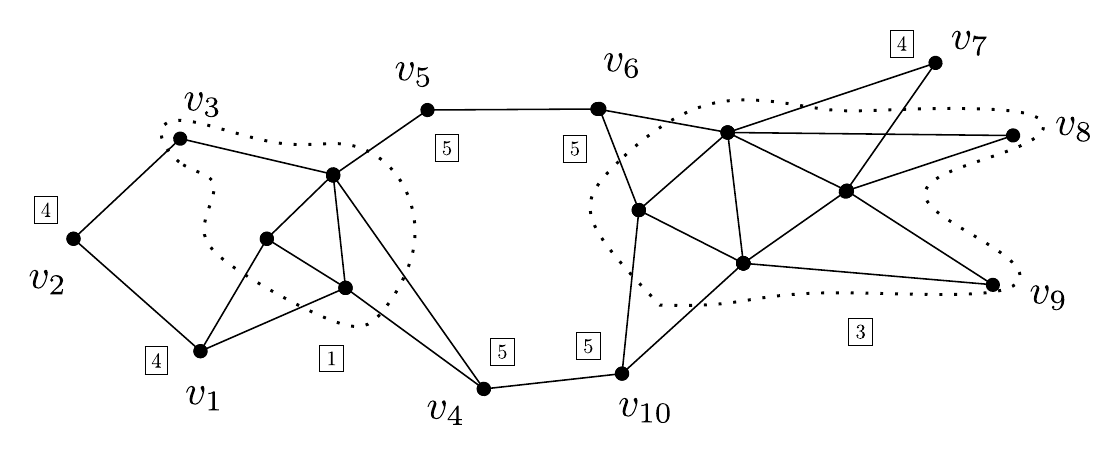}
\caption{Recoloring of the input graph in the second step of the algorithm.\label{fig:fpt-cactus-2}}
\end{figure}

\begin{enumerate}[resume]
\item {\it It refines coloring $f$ by recoloring some vertices using colors in 
$\{4, 5\}$ to ensure that any non-trivial connected component of $G[f^{-1}(i)]$ contains a big witness set.} 
\end{enumerate}
%For recoloring, the algorithm identifies some vertices that are not part of 
%any big witness set in any witness structure.
The algorithm recolors the vertices using the following two criteria.
First,  it recolors vertices that are only adjacent to vertices of different colors. 
For example,  vertices $v_1, v_2$ are recolored to $4$.
Second,  it recolors vertices on a path between 
two non-trivial monochromatic connected components.
For example,  vertices $v_4, v_5, v_6, v_{10}$ in Figure~\ref{fig:fpt-cactus-2}.
We prove in Lemma~\ref{lemma:recoloring-path-one-witness} and 
Lemma~\ref{lemma:recoloring-path-two-witnesses}, respectively that this recoloring is safe. 
in the sense
that there still remains a desirable $T$-witness structure of $G$ whose 
`big witness sets' (see Definition~\ref{def:big-witness-set}) are monochromatic.
We formalize this step in Section~\ref{sec:refinement-coloring}.

%Vertices of paths that do not intersect with any big witness set in any witness structure and are {adjacent} to only one big witness set (Lemma~\ref{lemma:recoloring-path-one-witness}) are colored using the color $4$.
%(Ex. vertices $v_1, v_2$ in Figure~\ref{fig:fpt-cactus}.) 
%Vertices that are not a part of any big witness set and lie on a path between two big witness sets (Lemma~\ref{lemma:recoloring-path-two-witnesses}) are recolored to color $5$.
%(Ex. vertices $v_4, v_5, v_6, v_{10}$ in Figure~\ref{fig:fpt-cactus-3}.)

\begin{enumerate}[resume]
\item {\it Finally,  it extracts a big witness set from the components of $G[f^{-1}(i)]$ for $i \in \{1, 2, 3\}$.}
\end{enumerate}
For a component $C$ of $G[f^{-1}(i)]$, the algorithm finds a subset of vertices 
that is `as good as' the big witness set in it by computing a \emph{connected core} (Definition~\ref{def:conn-core}).
Figure~\ref{fig:fpt-cactus-3} highlights the connected cores of color class with dotted lines.
We formalize this step in Section~\ref{sec:cactus-extract-big-witness}.
\begin{figure}[t]
\centering
\includegraphics[scale=0.65]{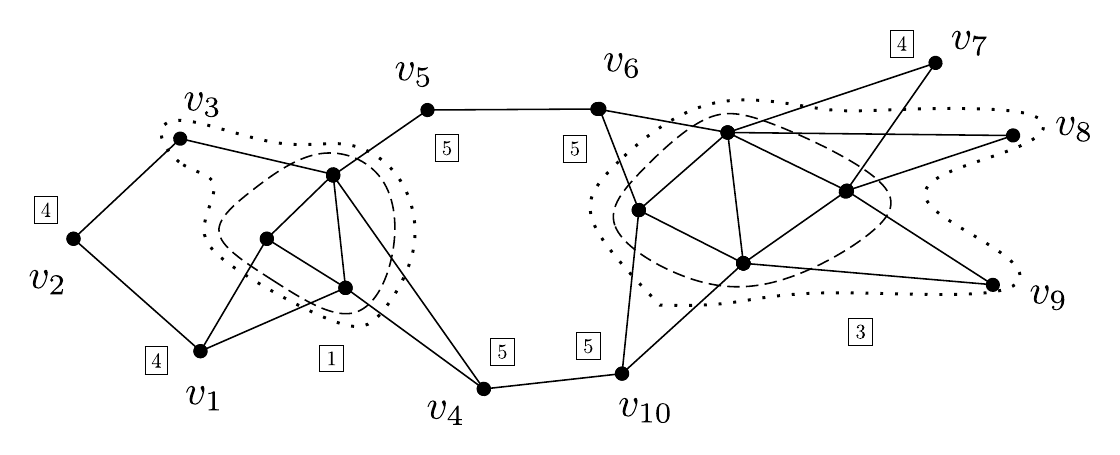}
\caption{Extracting a big witness set (denoted by dashed boundaries) from
a color class (denoted by dotted boundaries). \label{fig:fpt-cactus-3}}
\end{figure}

Now,  the algorithm considers the partition of $V(G)$ that contains connected cores
and singleton sets for the remaining vertices. 
Let $F$ be the collection of edges of the spanning tree of these connected cores.
If $|F| \le k$ and $G/F$ is a cactus, then the algorithm returns \yesalgo; otherwise, it returns \noalgo.
This completes the overview of the algorithm for $2$-connected graphs.

We present a randomized algorithm for general graphs that uses the above algorithm as a subroutine 
in Section~\ref{sec:cactus-algo-together}.
It also contains a way to its derandomization.
%to use the above algorithm to obtain an algorithm for general graphs.
%We prove that if $(G, k)$ is a \yesinst\ instance, the algorithm returns \yesalgo\ with constant probability.
%If $(G, k)$ is a \noinst\ instance, then the algorithm always returns \noalgo.
We remark that the primary goal of this work is to provide a single 
exponential-time algorithm for {\sc Cactus Contraction},  and
for the sake of simplicity, we do not optimize the running time. 

%\paragraph{Organization of the paper}
%Preliminary notation and results are presented in Section~\ref{sec:prelims}.
%The randomized algorithm for \textsc{Cactus Contraction} on $2$-connected graphs is described in Sections~\ref{sec:coloring}, \ref{sec:refinement-coloring} and \ref{sec:cactus-extract-big-witness}.
%The randomized algorithm for \textsc{Cactus Contraction} on general graphs and its derandomization using $(n,k)$-universal sets are presented in Section~\ref{sec:cactus-algo-together}.

\section{Preliminaries}
\label{sec:prelims}

For an integer $q$, we denote the set $\{1, 2, \dots, q\}$ by $[q]$. 

\subsection{Graph-theoretic Terminology}
All graphs considered in this paper are undirected, finite, and simple. 
For a graph $G$, $V(G)$ and $E(G)$ denote the sets of its vertices and edges, respectively. 
Two vertices $u, v$ are {\em adjacent} in $G$ if there is an edge $uv$ in $E(G)$. 
The {\em neighborhood} of a vertex $v$, denoted by $N_G(v)$, 
is the vertices adjacent to $v$ in $G$. 
The {\em degree} $d_G(v)$ of the vertex $v$ is $|N_G(v)|$. 
We extend the notion of the neighborhood for subsets $S$ as 
$N_G(S)=\bigcup_{v \in S} N_G(v) \setminus S$. 
For a subset $F$ of $E(G)$, $V_G(F)$ denotes the set of endpoints of edges in $F$. 
For subsets $X, Y$ of $V(G)$, $E_G(X, Y)$ denotes the set of edges with one endpoint in $X$
and the other endpoint in $Y$. 
Similarly, $E_G(X)$ denotes the set of edges with both endpoints in $X$. 
We omit the subscript in the notation for these sets if the graph under consideration is apparent. 
For a graph $G$ and an integer $q$, we say that the function $g: V(G) \rightarrow [q]$, is a \emph{proper coloring} of $G$ if for any edge $uv \in E(G)$, $g(u) \neq g(v)$.

For a subset $S$ of $V(G)$, $G - S$ denotes the graph obtained by 
deleting $S$ from $G$ which is also known as
\emph{induced subgraph} on $V(G) - S$.
We use $G[S]$ to denote the subgraph of $G$ induced on the set $S$. 
A {\em path} $P = (v_1,\ldots, v_q)$ is a sequence of distinct vertices such that 
any pair of consecutive vertices are adjacent. 
\begin{definition}[Cable Path]
A path $P = (v_1, v_2, \ldots, v_q)$ in $G$ is a \emph{cable path} 
if for each $ 2 \leq i \leq q - 1$,$N_G(v_i) = \{v_{i-1}, v_{i+1} \}$,  
i.e. ,  $v_i$ is adjacent with 
only two vertices in $G$ which are also on the path.
\end{definition}
The vertex set of $P$, denoted by $V(P)$, is the set $\{v_1,\ldots, v_q\}$. 
A {\em cycle} $C = (v_1,\ldots, v_q)$ is a sequence of distinct vertices
that is a path where $v_1$ and $v_q$ are adjacent. 

A graph is {\em connected} if there is a path between every pair 
of its vertices and it is {\em disconnected} otherwise.  
A subset $S$ of $V(G)$ is a \emph{connected set of vertices} if $G[S]$ is connected. 
A {\em component} of $G$ is a maximal connected subgraph of $G$. 
If the number of connected components of $G-\{v\}$ is more than that of $G$, then we call $v$ as a {\em cut-vertex}.
A connected graph with no cut-vertex is called a {\em 2-connected} graph. 
A \emph{block} is a  maximal subgraph that is $2$-connected. 
A block is a maximal induced $2$-connected subgraph, edge, or isolated vertex. 
Also, two distinct blocks intersect in at most one vertex. 
A vertex in at least two blocks is a \emph{cut-vertex}. 

A \emph{cactus} is a connected graph in which every edge is in at most one cycle. 
A cactus is a connected graph in which every block is either an edge or a cycle.
The following properties of cacti follow imminently from their definition. 

\begin{observation} \label{obs:cactus-prop}
The following statements hold for a cactus $T$.
\begin{enumerate}[noitemsep,nolistsep]
\item \label{item:cactus-coloring} The vertices of $T$ can be properly colored using at most three colors.
\item \label{item:high-deg} Every vertex of degree at least three in $T$ is a cut-vertex.
\item \label{item:subdivide} If $T’$ is the graph obtained from a cactus $T$ by subdividing an edge then $T’$ is a cactus. 
\end{enumerate} 
\end{observation}

\subsection{Graph Contraction and Related Terms}
\label{sub-sec:graph-contraction}

The {\em contraction} of the edge $e=uv$ in $G$ deletes vertices 
$u$ and $v$ from $G$, and adds a new vertex adjacent to vertices 
that were adjacent to either $u$ or $v$. 
This process does not introduce self-loops or parallel edges and 
$ G/e$ denotes the resulting graph. 
That is, $V(G/e) = (V(G) \cup \{w\}) \backslash \{u, v\}$ and 
$E(G/e) = \{xy \mid x,y \in V(G) \setminus \{u, v\}, xy \in E(G)\} 
\cup \{wx \mid x \in N_G(u) \cup N_G(v)\}$ where 
$w$ is a new vertex not in $V(G)$.

\begin{definition}[Graph Contraction] \label{def:graph-contractioon} 
A graph $G$ is said to be \emph{contractible} to graph $H$ if there is a surjective function 
$\psi: V(G) \rightarrow V(H)$ such that
\begin{itemize}
\item[-] for any vertex $h \in V(H)$,  the set $W(h) := \{v \in V(G) \mid \psi(v)= h\}$ is connected in $G$, and
\item[-] any two vertices $h, h' \in V(H)$ are adjacent in $H$ if and only if there is an edge in $G$ with one endpoint in $W(h)$ and the other in $W(h’)$.
\end{itemize}
We say that $G$ is contractible to $H$ via mapping $\psi$ and define $H$-\emph{witness structure} $\calW$ of $G$ as $\mathcal{W} := \{W(h) \mid h \in V(H)\}$.
\end{definition}
Observe that $\mathcal{W}$ is a partition of vertices in $G$. 
We can obtain $H$ from $G$ by sequentially contracting edges in a spanning forest, say $F$,  of big witness sets in $H$-witness structure $\calW$ of $G$.
For a vertex $h$ in $H$,  $W(h)$ is a \emph{witness set} associated with or corresponding to $h$.  

\begin{definition}
\label{def:big-witness-set}
A witness set $W(h)$ in $H$-witness structure of $G$ is said to be \emph{big witness set} if it contains at least two vertices.
\end{definition}

Note that there are at most $|F|$ big witness sets, and the number of vertices in a big witness set is upper bounded by $|F| + 1$. 
We say $G$ is \emph{$k$-contractible} to $H$; equivalently,  $H$ can be obtained from $G$ by at most $k$ edge contractions if the cardinality of $F$ is at most $k$.

For a given subset of edges $F \subseteq E(G)$, 
we consider the partition,  say $\calW$,  of $V(G)$ 
that contains vertices in each connected component  of $G[V(F)]$
as a separate part
and singleton set for vertices in $V(G) \setminus V(F)$.
Given this partition, we can construct graph $H$ by adding a vertex for each part and an edge in $H$ if these two parts are adjacent in $G$.
We denote this graph $H$ by $G/F$.
Consider a set of edges $F' \subseteq F$ such that $F'$ is a collection of edges in a spanning forest 
of $G[V(F)]$.
Note that the graphs constructed above with respect to $F$ and $F'$ are identical.
Hence, it is safe to assume that edges in $F$ constitute a forest.
Finally, we mention the following useful observation to end this section.

\begin{observation}
\label{obs:cactus-witness-cut-vertex}
Consider a $2$-connected graph $G$ and let $F$ be a 
set of edges in $G$ such that $G/F$ is a cactus.
If $t$ is a cut vertex in $G/F$, the witness set $W(t)$ in 
the $G/F$-witness structure contains at least two vertices.
\end{observation}
\begin{proof} 
For the sake of contradiction, assume that $t$ is a cut-vertex in $T = G/F$ and $W(t)$ is a singleton set in $T$-witness structure.
Let $W(t) = \{u\}$.
We argue that $u$ is a cut-vertex in $G$.
Let $T_1$ and $T_2$ be any two connected components obtained by removing $t$ from $T$.
Consider the set $V_1$, the collection of vertices present in witness sets corresponding to vertices in $T_1$. 
Formally, $V_1 = \{u |\ u \in W(t_1) \text{ for some } t_1 \in V(T_1)\}$.
Similarly, $V_2 = \{u |\ u \in W(t_2) \text{ for some } t_2 \in V(T_2)\}$.
Since $T_1, T_2$ are non-empty, so are $V_1, V_2$. 
Further,  there is no edge between $T_1, T_2$ in $T$.
Moreover, since $T$ is obtained from  $G$ by contracting edges, there is no edge between $V_1, V_2$ in $G$.
This implies that $G - v$ has at least two connected components viz $V_1, V_2$.
This contradicts the fact that $G$ is a $2$-connected graph. 
Hence for every cut vertex $t$ in $T$, the associated witness set $W(t)$ contains at least two vertices.
\end{proof}

\subsection{Parameterized Complexity}
An instance of a parameterized problem comprises an input $I$ 
of the classical instance of the problem and an integer $k$, 
which is called the {\em parameter}. 
A problem $\Pi$ is said to be \emph{fixed-parameter tractable} 
or in the complexity class \FPT\ if given an instance $(I,k)$ of $\Pi$, 
we can decide whether or not $(I,k)$ is a \yesinst\ instance of $\Pi$ 
in $f(k)\cdot |I|^{\OO(1)}$ time. 
Here, $f(\cdot)$ is some computable function that depends only on $k$. 
We say that two instances, $(I, k)$ and $(I’, k’)$, of a parameterized problem 
$\Pi$ are \emph{equivalent} if $(I, k)$ is a \yesinst\ instance of $\Pi$ 
if and only if $(I’, k’)$ is a \yesinst\ instance of $\Pi$. 
A \emph{reduction rule}, for a parameterized problem $\Pi$ is an algorithm 
that takes an instance $(I, k)$ of $\Pi$ as input and outputs an instance $(I’, k’)$
of $\Pi$ in time polynomial in $|I|$ and $k$.
If $(I, k)$ and $(I’, k’)$ are equivalent instances then we say that 
the reduction rule is \emph{safe} or \emph{correct}.
For more details on parameterized complexity, we refer to the books by Cygan et al.~\cite{cygan2015parameterized} and Fomin et al.~\cite{fomin2019kernelization}.

\section{Compatible Colorings}
\label{sec:coloring}

In this section,  we introduce the notion of a \emph{compatible coloring} of vertices of $G$
with respect to a $T$-witness structure of $G$.
Informally, a compatible coloring $f_c$ assigns the same color to every vertex in a big witness set, allowing us to define the color of a witness set. 
Further, if two big witness sets are adjacent, 
then $f_c$ colors them differently.  
Finally,  if there is a cable path between two big witness sets with 
internal vertices corresponding to singleton witness sets, then $f_c$ colors 
the endpoints of this path with a color different from the color of their neighbors in the path. 
See Figure~\ref{fig:compatible-coloring} for an illustration.
We then prove that if the witness structure corresponds to a cactus,  
the input graph $G$ admits a compatible coloring (Lemma~\ref{lemma:3-col-existence}). 
Then, we show that if $G$ is a $2$-connected graph that is $k$-contractible to a cactus $T$,  
a random coloring of its vertices is a coloring compatible with a fixed $T$-witness structure 
with high probability (Lemma~\ref{lemma:random-coloring-compatible}).

\begin{definition}[Compatible Coloring] \label{def:compatible} 
Consider a graph $G$ and its $T$-witness structure $\calW$.
A coloring $f_c: V(G) \mapsto \{1, 2, 3\}$ is a {\em compatible coloring} with respect to $\mathcal{W}$ 
if the following hold.
  \begin{enumerate}
  \item \label{item:mono-big-witness} For any $W(t) \in \mathcal{W}$,  every vertex in $W(t)$ has same color.
  Hence,  $f_c(W(t))$ is well-defined.
  \item \label{item:adj-big-witness} For any edge $t_xt_y \in E(T)$,  if $W(t_x), W(t_y)$ are big witness sets 
  then $f_c(W(t_x)) \neq f_c(W(t_y))$.
  \item \label{item:path-big-witness} For any cable path $P = (t_x, t_1, t_2, \dots, t_q, t_y)$ in $T$ 
  such that $W(t_x),  W(t_y)$ are big witness sets and 
  for all $1 \le i \le q$,  $W(t_i)$ is a singleton witness set 
  then $f_c(W(t_x)) \neq f_c(W(t_1))$ and $f_c(W(t_y)) \neq f_c(W(t_q))$. 
  \end{enumerate}
\end{definition}

A subset $X$ of $V(G)$ is called a \emph{monochromatic component} 
of $f_c$, 
if $X$ is an inclusion-wise maximal connected set of vertices that 
have the same color in $f_c$. 
Let $\calX$ be the set of all monochromatic components of $f$. 
Note that for any monochromatic  component $X$ in $\calX$, 
either all vertices of $X$ are in singleton witness sets in $\calW$ 
or $X$ contains exactly one big witness set in $\calW$. 

Figure~\ref{fig:compatible-coloring} presents an illustrative example. 
The big witness sets $W(t_x), W(t_y), W(t_z)$ are monochromatic. 
The vertices $t_x$ and $t_y$ are adjacent and $W(t_x)$ and $W(t_y)$ 
are big witness sets and hence have different colors. 
Consider the path between $W(t_x)$ and $W(t_z)$ whose 
internal vertices correspond to singleton witness sets. 
The coloring of vertices in this path satisfies the definition's third property. Finally, consider the path which starts and ends 
in $W(t_x)$ and all internal vertices corresponding to singleton witness sets. 
The definition of compatible coloring allows all vertices in this path 
to have the same color as $W(t_x)$. 

\begin{figure}
\centering 
\includegraphics[scale=0.7]{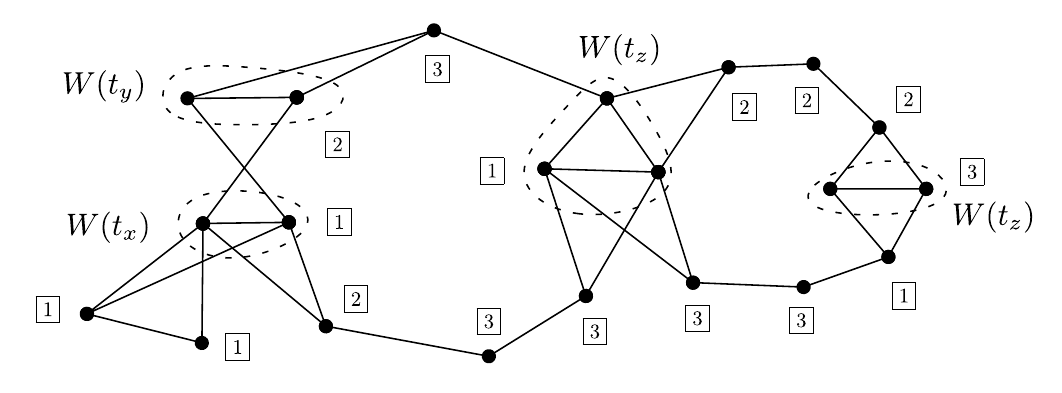}
\caption{A compatible coloring of the input graph. 
Dotted boundaries denote big witness sets.}
\label{fig:compatible-coloring}
\end{figure}

Note that we use compatible coloring using 
only three colors in this example. 
In the following lemma, we argue three colors 
are sufficient to get such a coloring if the resultant graph, i.e., 
$G/F = T$ is a cactus. 

\begin{lemma} \label{lemma:3-col-existence}
Consider a graph $G$ and a set $F \subseteq E(G)$ such that 
$G/F$ is a cactus. Let $\calW$ be the $G/F$-witness structure of $G$. 
Then, there exists a coloring $f : V(G) \rightarrow \{1, 2, 3\}$ which is 
compatible with $\calW$.
\end{lemma}
\begin{proof}
Let $G$ be contractible to $G/F$ via function $\psi$. 
That is, the  function $\psi: V(G) \rightarrow V(G/F)$ satisfies 
the properties mentioned in Definition~\ref{def:graph-contractioon}. 
As $G/F$ is a cactus, by Observation~\ref{obs:cactus-prop}~($\ref{item:cactus-coloring}$), 
there is a proper coloring  $g: V(G/F) \rightarrow \{1, 2, 3\}$ of $G/F$.
We define the coloring $f_c: V(G) \rightarrow \{1, 2, 3\}$ as: $f_c(v) = g(\psi(v))$ 
for every $v$ in $V(G)$.
It is easy to verify that the coloring $f_c$ satisfies all the properties 
mentioned in Definition~\ref{def:compatible}.
This concludes the proof of the lemma.
\end{proof}

Next, we argue that random coloring is a compatible coloring
with respect to a fixed witness structure $\calW$ with high probability. 
We use the following claim to bound the number of vertices in $G$
which are adjacent to big witness sets, are in singleton witness sets, 
and are in a path between two big witness sets.  

\begin{lemma}\label{lemma:nbd-cut-cactus} 
Consider a $2$-connected graph $G$ which is contractible to a 
cactus $T$. 
Let $T_B$ be the set of vertices in $T$ that correspond to 
big-witness sets. 
Then, there are at most $4 \cdot |T_B|$ vertices lie on a path between two different vertices in $T_B$ and are adjacent 
to vertices in $T_B$.
\end{lemma}
\begin{proof} 
Let $F \subseteq E(G)$ denote a set of edges such that $G/F = T$. 
Then,  $T_B = \{t' \mid t' \in V(T)$ such that $W(t')$ is a big witness set
in $G/F$-witness structure of $G\}$. 
Let $V^{\circ}$ be the set of vertices $t$ in $V(T)$ such that 
$t$ is in a path between two different vertices,  say $t_x,  t_y$,  in $T_B$
and is adjacent with $t_x$ and $t_y$.
To prove the claim, we need to bound the cardinality of $V^{\circ}$.

Note that if $T_B$ contains at most one vertex, then by the definition
of $V^{\circ}$, we have $|V^{\circ}| = 0$ and so the statement is 
vacuously true. 
We consider the case when $|T_B| \ge 2$ and prove that 
$|V^{\circ}| \le 4 \cdot |T_B|$ using induction on the number 
of blocks in $T$. 
Consider the base case when $T$ has only one block. 
In this case, $T$ is an edge or a cycle.  
It is easy to see that every vertex in $T_B$ can contribute at most 
two vertices to $V^{\circ}$. 
Hence, $|V^{\circ}| \le 2 \cdot |T_B|$. 
As $|T_B| \ge 2$, we have $|V^{\circ}| \le 4 \cdot (|T_B| - 1)$. 
Our induction hypothesis is as follows: if the number of blocks in $T$
 is strictly less than $q \geq 1$ then $|V^{\circ}| \le 4 \cdot |T_B|$. 
 Suppose that $T$ that has $q$ blocks. 
 Let $\mathcal{D}$ be the block decomposition of $T$. 
 Consider a leaf in $\mathcal{D}$ and let $D$ be the corresponding block. 
Let $t$ be the unique cut vertex in $D$.
Consider the subgraphs $T_1, T_2$ of $T$ induced on 
$V(T) \setminus (D \setminus \{t\})$ and $D$, respectively.
It is easy to verify that $T_1, T_2$ are cactus graphs. 
Note that $T^1 \cap T^2 = \{t\}$ and $|T| = |T^1| + |T^2| - 1$.
Similarly, define $V^{\circ}_1 = V^{\circ} \cap V(T_1)$ and 
$V^{\circ}_2 = V^{\circ} \cap V(T_2)$.
By Observation~\ref{obs:cactus-witness-cut-vertex},  
$W(t)$ is a big witness set as $t$ is a cut-vertex in $T$.
This implies that  $t \not\in V^{\circ}$ and hence 
$V^{\circ}_1 \cap V^{\circ}_2 = \emptyset$ and 
$|V^{\circ}| = |V^{\circ}_1| + |V^{\circ}_2|$.

Define $T^1_B = T_B \setminus (D \setminus \{t\})$ and 
$T^2_B = T_B \cap D$.  Since $T_1$ has $q - 1$ blocks in its 
block decomposition, by the induction hypothesis, 
$|V^{\circ}_1| \le 4 \cdot |T^1_B|$.  
By the arguments in the base case,  
we have $|V^{\circ}_2| \le 4 \cdot (|T^2_B| - 1)$. 
Using the bounds on $|V^{\circ}|$ and $|T_B|$,  
we get $|V^{\circ}| \le 4 \cdot |T_B|$. 
This concludes the proof of the lemma.
\end{proof}

We are now in a position to argue that random coloring is 
compatible coloring with respect to a fixed witness structure 
$\calW$ with high probability.

\begin{lemma}\label{lemma:random-coloring-compatible} 
Consider a $2$-connected graph $G$ which is $k$-contractible 
to a cactus $T$.
Let $f: V(G) \rightarrow \{1, 2, 3\}$ be a random coloring where 
colors are chosen uniformly at random for each vertex.
Then, $f$ is compatible with the $T$-witness structure of $G$ 
with probability at least $1\large/{3^{6k}}$.
\end{lemma}
\begin{proof}
Let set $S \subseteq V(G)$ be the collection of all vertices 
in $V(G)$ which are in big witness sets in $\calW$ or 
are adjacent to a big witness set and are in paths between 
two big witness sets.
Since $G$ is $k$-contractible to $T$,  
there are at most $k$ big witness sets.
This implies the total number of vertices in $S$ is at 
most $2k + 4k = 6k$ (by Lemma~\ref{lemma:nbd-cut-cactus}).
We can ensure $|S| = 6k$ by arbitrarily adding some extra vertices.
By the definition of compatible coloring (Definition~\ref{def:compatible}), 
to determine whether a random coloring $f_c$ is compatible with 
$\calW$ or not, 
we only need to check the color of vertices in $S$.

Let $f_c$ be a $3$-coloring of $G$ which is compatible with $\calW$.
By Lemma~\ref{lemma:3-col-existence}, such a coloring exists.
For a random coloring $f$ and a vertex $v$ in $S$, 
the probability that $f_c(v) = f(v)$ is $1/3$.
Since colors are chosen uniformly at random for each vertex while constructing $f$, 
the probability that $f$ and $f_c$ color $S$ identically is at least $1\large/3^{6k}$. 
Hence $f$ is compatible with $\calW$ with probability at least $1\large/{3^{6k}}$.
\end{proof}

\section{Coloring Refinement} 
\label{sec:refinement-coloring}

Suppose $G$ is a 2-connected graph that is $k$-contractible
to a cactus $T$ and $F \subseteq E(G)$ is a minimal collection 
of at most $k$ edges such that $G/F = T$. 
Let $\calW$ denote the $T$-witness structure of $G$ and 
$f$ be a compatible coloring of $G$ with respect to $\calW$. 
A cycle in $T$ is called a \emph{pendant cycle} if there is exactly
one vertex in the cycle corresponding to a big witness set.
We now describe a recoloring procedure that expands the range 
of $f$ from $\{1, 2, 3\}$ to $\{1, 2, 3, 4, 5\}$ by identifying \emph{some} 
vertices that do not belong to any big witness set in $\calW$ and 
recoloring them with color $4$ or $5$. 
Note that algorithm does this recoloring without the knowledge of $\calW$. 
Let $\calX$ be the set of all monochromatic components with respect to $f$.

\begin{recoloring} \label{recolor:path-one-witness}
For any monochromatic component $X$ in $\calX$, if $G - X$ 
contains a single vertex or a cable path in $G$ as one of its 
components, then recolor all the vertices in that component with color $4$. 
\end{recoloring}

\begin{recoloring}\label{recolor:path-two-witnesses}
For any two monochromatic components $Y, Z$ in $\calX$, 
if $G - (Y \cup Z)$ contains a single vertex or a maximal cable 
path in $G$ as one of its components, then 
recolor all the vertices in that component with color $5$. 
\end{recoloring}

The algorithm exhaustively applies 
Recoloring~\ref{recolor:path-one-witness} before starting 
Recoloring~\ref{recolor:path-two-witnesses}. 
Also, once it recolors a vertex to $4$, it does not change its color. 
We argue that the algorithm recolors the vertices outside any big witness set in $\calW$ (the witness structure we are interested in). 
Then, we prove the following additional property of recoloring in 
Subsection~\ref{subsec:exact-one-big-witness}: any nontrivial monochromatic component of $f$ colored with colors $1, 2$, or $3$ 
contains exactly one big witness set of $\calW$.

Now, we present a brief rationale behind the recoloring rules. 
Consider a pendant cycle $C_T$ in $T$ such that $t$ is the unique 
cut vertex in $C_T$ and all vertices $C_T\setminus \{t\}$ correspond 
to singleton witness sets. 
Let $X$ be the monochromatic component 
which contains $W(t)$. 
Let $V_1$ be the collection of vertices in $G - X$ that are in singleton witness sets 
corresponding to vertices in $C_T$. 
It is easy to see that $V_1$ induces a cable path in $G$. 
Recoloring~\ref{recolor:path-one-witness} identifies such paths. 
Next, consider a cable path $P = (t_x, t_1, t_2, \dots, t_q, t_y)$ in $T$ 
such that $W(t_i)$ is a singleton witness set for all $1 \leq i \leq q$, 
and $W(t_x), W(t_y)$ are big witness sets. 
Let $X, Y$ be the monochromatic components containing 
$W(t_X), W(t_Y)$, respectively. 
Since $f$ is compatible with $\calW$,  we know that $t_1, t_q$ are not 
in $X$ and $Y$, respectively. 
Let $V_1$ be the vertices in $G$ which are contained in $W(t_i)$ 
for $1 \leq i \leq q$. It is easy to see that $V_1$ induces a maximal 
cable path in $G$. 
Moreover, $V_1$ is a component of $G - (X \cup Y)$. 
Recoloring~\ref{recolor:path-two-witnesses} identifies such paths. 

For example, in Figure~\ref{fig:fpt-cactus-2}, the path $v_1v_2$ is a 
component of $G - X$ where $X$ is a monochromatic component 
with color $1$. 
Recoloring~\ref{recolor:path-one-witness} changes the color of 
$v_1v_2$ to $4$. 
Similarly, $v_7$ is recolored by Recoloring~\ref{recolor:path-one-witness}. 
Note that the algorithm cannot identify $v_3$ or $v_8$ in this step. 
The path $v_5v_6$ is a maximal path when two monochromatic components are deleted from the graph. 
Recoloring~\ref{recolor:path-two-witnesses} 
changes the color of these two vertices to $5$.

\subsection{Properties of Recoloring~\ref{recolor:path-one-witness}}
\label{subsec:recoloring-simple-path}

Consider a monochromatic component $X \in \calX$. 
Let $P$ be a component of $G - X$ where $P$ is a cable path in $G$. 
In Lemma~\ref{lemma:recoloring-path-one-witness}, we argue that all 
vertices in $V(P)$ are singleton witness sets in $\calW$ or that 
we are dealing with a simple instance mentioned in Lemma~\ref{lemma:cycle_instance}.

\begin{figure}[t]
  \centering
  \includegraphics[scale=0.7]{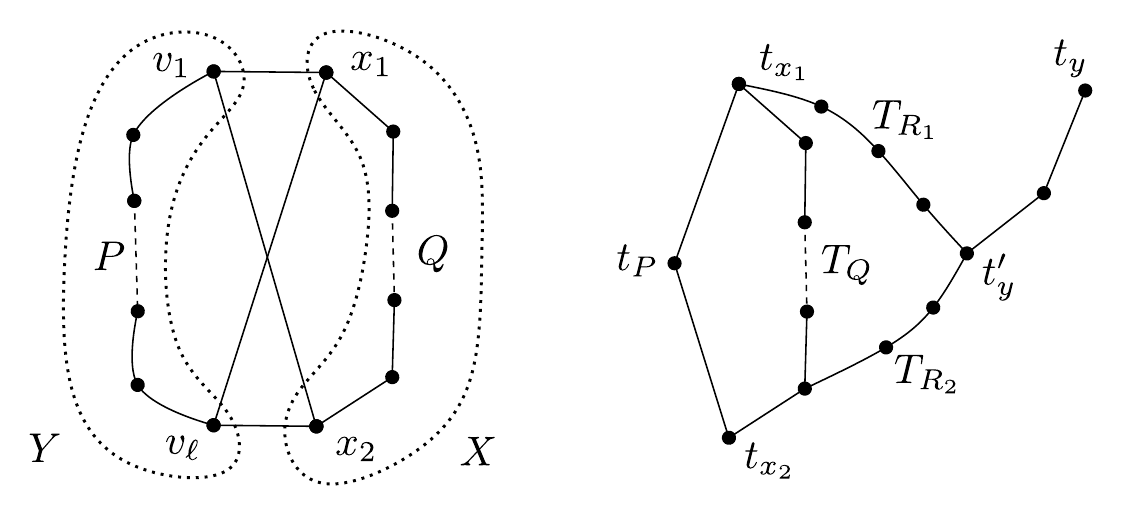}
  \caption{Illustration for Lemma~\ref{lemma:recoloring-path-one-witness}\label{fig:recoloring1}}
\end{figure}

\begin{lemma}\label{lemma:cycle_instance}
If $G$ is a $2$-connected graph such that $V(G)$ can be partitioned
into the vertex sets of two cable paths $P$ and $Q$ in $G$, then we can
solve the instance $(G,k)$ of \textsc{Cactus Contraction} in polynomial time.
\end{lemma}
\begin{proof}
Let $p_1,p_2$ and $q_1, q_2$ be the endpoints of the cable paths 
$P$ and $Q$, respectively. Observe that $G$ has a Hamiltonian cycle, 
as $G$ is $2$ connected and $p_1,p_2,q_1,q_2$ are the only vertices 
that may have degrees greater than two.
If $G$ is an induced cycle, then 
the optimal solution is the empty set. 
Otherwise, $G$ is a cycle with at most four additional edges between 
$p_1,p_2$ and $q_1,q_2$. 
Any optimal solution requires 
at most $3$ edge contractions and can be obtained in polynomial time.
\end{proof}

Subsequently, we assume that the instance under consideration does 
not satisfy the premise of Lemma~\ref{lemma:cycle_instance}. 

\begin{lemma}\label{lemma:recoloring-path-one-witness}  
For a monochromatic component $X$ in $\cal X$, let $P$ be a component
of $G-X$. 
If $P$ is a cable path in $G$, then all the vertices of $P$ lie in singleton witness sets.
\end{lemma}
\begin{proof}
For the sake of contradiction, assume that the lemma is false. 
Some big witness set in $\cal W$ contains a vertex of $P$.
Let $Y \in \calX$ be the monochromatic component that contains this witness set. 
As $f$ is a compatible coloring, and $N_G(P) \subseteq X$, we have $Y \subseteq V(P)$. 
Hence $Y$ is a cable path in $G$. 

We show that $Y = V(P)$. Let $P = (v_1, v_2, \ldots, v_\ell)$. 
Suppose $Y$ is a proper subset of $V(P)$. Then at least one of $v_1$ or 
$v_\ell$ is not in $Y$. 
Without loss of generality, let $v_1 \not\in Y$. 
Let $v_{i+1}$ be the smallest indexed vertex in $Y$. 
Let $t_i$ be the vertex in $T$ such that $v_i \in W(t_i)$. 
Observe that $W(t_i) \subseteq P$. 
There is no edge between $Y$ and $W(t_i)$ except for $v_i v_{i+1}$. 
Consider the witness structure $\cal W’$ obtained from $\cal W$ 
by replacing $Y$ with $Y \setminus \{v_{i+1}\}$ and $\{v_{i+1}\}$. 
Let $T’$ be the graph obtained from $G$ by contracting all witness 
sets in $\calW’$.
In other words, $\calW’$ is a $T’$-witness structure of $G$. 
Note that $T’$ can be obtained from $T$ by subdividing $t_i t_Y$, 
where $W(t_Y) = Y$. 
Hence, by  Observation~\ref{obs:cactus-prop},  $T’$ is a cactus. 
This contradicts the minimality of the solution $F$ associated 
with $\calW$. 
Therefore, $Y = V(P)$. 
  
All the edges of $P$ are in $F$, which implies that $P \in \calX$. 
Let $t_P$ be the vertex corresponding to $P$ in $T$ and $t_P$ 
is adjacent to a vertex $t \in T$ if and only if $W(t)$ contains a vertex 
from $N_G(P)$ which is a subset of $X$. 
We consider two cases depending on the number of edges across 
$V(P)$ and $X$. 
If $|E_G(V(P),X)| \leq 2$, then $|E_G(V(P),X)| = 2$ as $G$ is 2-connected, 
$P$ is a cable path in $G$ and $P$ is a component of $G-X$. 
Then, $t_P$ is adjacent with either one vertex, say $t_i$, or two vertices 
$t_i, t_j$ in $T$. 
By  deleting a vertex from $W(t_P)$ and adding 
it as a singleton witness set, we get a $T'$-witness structure 
$\mathcal{W}'$ where $T'$ is obtained from $T$ by subdividing 
the edge $t_it_P$. 
By Observation~\ref{obs:cactus-prop}, $T'$ is a cactus and 
this contradicts the minimality of $F$. 
Next, we consider the case when $|E_G(P,X)| \geq 3$. 
By the definition of a compatible coloring,  $X$ contains at most one 
big witness set. 
We consider two cases depending on whether $X$ contains a big 
witness set or not. 

\paragraph{Case~$1$: $X$ does not contain a big witness set}
Let $T_X$ denote the vertices in $T$ which correspond to singleton 
witness sets containing vertices in $X$. 
If $N_G(P)$ corresponds to at least $3$ vertices in $X$, then 
$T[T_X \cup \{t_P\}]$ contains two cycles with a common edge, 
i.e., $T$ is not a cactus leading to a contradiction.
Hence, $N_G(P)$ contains exactly two vertices, $x_1$ and $x_2$,
of $X$ and $E_G(P,X)$ contains either $3$ or $4$ edges. 
Refer to Figure~\ref{fig:recoloring1} for an illustration.

By Observation~\ref{obs:cactus-witness-cut-vertex}, 
no vertex of $T_X$ is a cut-vertex in $T$. 
As $X$ is a connected set, there is a path, say $Q$, 
between $x_1$ and $x_2$ which is in $X$. 
Let $T_Q$ denote the vertices in $T$, corresponding 
to singleton witness sets containing $Q$. 
Observe that $C = T[T_Q \cup \{t_P\}]$ is a cycle in $T$ with 
$t_P$ being the only vertex corresponding to a big witness set in $C$. 
We claim that there is no other vertex in $T$ apart from vertices $C$ 
i.e. $T = C$. 
If this is true, then $G = P \uplus Q$ and $P$ and $Q$ are cable paths in $G$. 
However,  this instance satisfies the premise of Lemma~\ref{lemma:cycle_instance}. 

We now argue that $T = C$. Assume this is not the case, 
then $V(T) \setminus (\{t_P\} \cup V(Q))$ is non-empty. 
We now show that there is a vertex $t_y \in V(T) \setminus (V(P) \cup V(Q))$, 
such that there are two internal vertex disjoint paths between 
$t_y$ and $t_P$ in $T$. 
Choose an arbitrary $t_y$ and consider a minimum separator between 
$t_y$ and $t_P$ in $T$. 
If the minimum separator is a single vertex $t’_y$, 
then observe that $y’ \notin V(Q)$, as vertices of $Q$ are not 
cut-vertices in $T$. 
We can substitute $t_y$ with $t’_y$ and start over. 
Since the shortest path between $t_y’$ and $t_P$ in $T$ 
is strictly shorter than the shortest path between $t_y$ and $t_P$, 
we obtain the required vertex $t_y$. 
Refer to Figure~\ref{fig:recoloring1} for an illustration. 
Let $T_{R_1}$ and $T_{R_2}$ be two internally vertex-disjoint paths 
in $T$ between $t_y$ and $t_P$. 
Without loss of generality, let $t_{x_1} \in T_{R_1}$ and $t_{x_2} \in T_{R_2}$, 
and hence $T_{R_1} \cup T_{R_2}$ contains a path between $t_{x_1}$ 
and $t_{x_2}$, say $T_R$, in $T$. 
This path is distinct from $T_Q$, as $Q_1 = V(T_Q) \cap V(T_{R_1})$ 
and $Q_2 = V(T_Q) \cap V(T_{R_2})$ are disjoint and therefore at 
least one edge of $T_Q$ is not in $T_R$. This implies that $T$ contains 
three distinct paths between $t_{x_1}$ and $t_{x_2}$, namely 
$P_T = (x_1, t_P, x_2), T_Q$ and $T_R$.  
This contradicts the fact that $T$ is a cactus. 
Hence, we conclude that $T = C$.
    
\paragraph{Case~2. $X$ contains a big witness set}
Let $Z$ be the big witness set in $X$ and 
$t_Z$ be the vertex in $T$ obtained by contracting $Z$. 
We claim that $N_G(P)$ is a subset of $Z$. 
If not, consider a vertex $v$ in $N_G(P) \setminus Z$. 
Note that $v$ is in a singleton witness set, say $W(t_v)$.
In $T$, $t_v$ lies on a path between $t_Z$ and $t_P$.  
Now $t_Z$, $t_P$ are big witness sets with $f(t_v) = f(t_Z)$ 
contradicting the fact that $f$ is a compatible coloring 
(Definition~\ref{def:compatible} (\ref{item:path-big-witness})). 
Hence $N_G(P) \subseteq Z$ which implies that $N_T(t_P) = \{t_Z\}$ in $T$. 
Then, $G/(F \setminus E(P))$ is also a cactus contradicting the minimality of $F$.
\end{proof}

Note that the above lemma holds when $P$ contain only one vertex. 
For a monochromatic component $X$ in $\calX$, let an isolated vertex 
$v$ be a component of $G-X$. 
Observe that $v$ has a color that is different from $X$. 
Since $f$ is compatible with an optimum solution, all big witness sets are monochromatic. 
This implies that $v$ cannot be in any big witness set and 
remains as a singleton witness set.

\subsection{Properties of Recoloring~\ref{recolor:path-two-witnesses}}
\label{subsec:recoloring-pendent-cycles}

Recall that in a cable path, no internal vertex is adjacent to any vertex 
outside this path. 
A cable path is \emph{maximal} if not contained in any other cable path. 
Since we are working with a $2$-connected graph, in a \emph{maximal cable path}
every internal vertex has the degree of exactly two, and the endpoints have the degree 
strictly greater than two. We state the following lemma when $P$ is a maximal 
cable path, but it also holds when $P$ is a single vertex. 
We use the maximality of $P$ to prove the second part of the lemma.

\begin{lemma} \label{lemma:recoloring-path-two-witnesses} 
For two monochromatic components $Y, Z$ in $\cal X$, let $P$ be a component 
of $G-(Y \cup Z)$. 
If $P=(v_1, v_2, \ldots, v_\ell)$ is a maximal cable path in $G$ with 
$N_G(v_1) \subseteq Y \cup \{v_2\}$ and $N_G(v_\ell) \subseteq Z \cup \{v_{\ell -1}\}$, 
then all vertices of $P$ lie in singleton witness sets. 
Further, both $Y$ and $Z$ contain big witness sets.
\end{lemma}
\begin{proof}
For the sake of contradiction, assume that
there is a big witness set that contains a vertex of $P$. 
Let $A \in \calX$ be a monochromatic component that contains 
this witness set. 
As $f$ is a compatible coloring and $N_G(P) \subseteq Y \cup Z$, 
we have $A \subseteq P$. 
This implies that $A$ is a cable path in $G$. 
Since $A$ contains a big witness set, $A$ itself is a big witness set. 
We now argue that $A = V(P)$. 
Suppose that $A$ is a proper subset of $V(P)$. 
Then at least one of $v_1, v_\ell$ is not in $Y$. 
Without loss of generality, let $v_1 \not\in A$. 
Let $v_{i+1}$ be the smallest indexed vertex in $A$ and 
let $t_i$ be the vertex in $T$ such that $v_i \in W(T_i)$. 
Observe that $W(t_i) \subseteq P$.  
There is only one edge, $v_iv_{i+1}$, between $A$ and $W(t_i)$. 
Consider the witness structure $\cal W’$ obtained from $\cal W$ 
by replacing $A$ with $A \setminus \{v_{i+1}\}$ and $\{v_{i+1}\}$. 
Let $T’$ be the graph obtained from $G$ by contracting all witness 
sets in $\calW’$. 
Note that $T’$ can be obtained from $T$ by sub-dividing $t_i t_Y$, 
where $W(t_Y) = Y$. 
Hence by Observation~\ref{obs:cactus-prop},  
this contradicts the minimality of the solution associated with $\calW$.

All the edges of $P$ are in $F$, which implies that $P \in \calX$. 
Let $t_P$ be the vertex corresponding to $P$ in $T$. 
Let $YP = Y \cap N(P) = Y \cap N(v_1)$ and 
$ZP = Z \cap N(P) = Z \cap N(v_\ell)$. 
First, we claim that $YP$ is in a witness set of $\cal W$. 
Suppose that $W(t_1), W(t_2)$ are in two different witness sets 
in $\cal W$ which contains vertices from $YP$. 
Refer to Figure~\ref{fig:recoloring2} for an illustration.

\begin{figure}[h]
  \centering
  \includegraphics[scale=0.6]{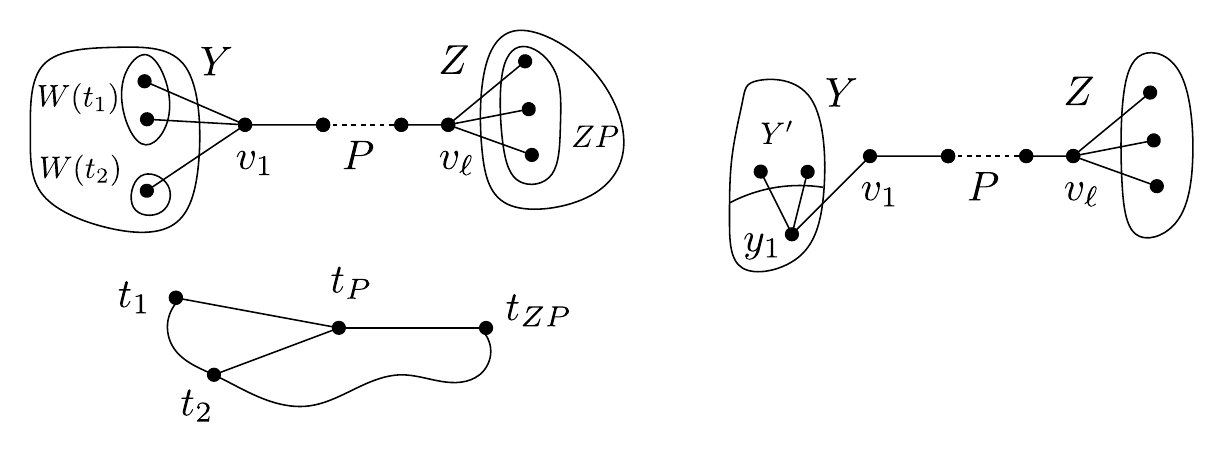}
  \caption{Illustration for Lemma~\ref{lemma:recoloring-path-two-witnesses}. \label{fig:recoloring2}}
\end{figure}

As $G$ is a $2$-connected graph and $P$ is a cable path in $G$, 
$G - V(P)$ is connected. 
As $G$ is a $2$-connected graph, there exists a path, say $P_1$, 
between $YP$ and $ZP$ which does not contain $v_1$.  
Since $v_1$ is in the cable path $P$, 
path $P_1$ does not contain 
any vertex in $V(P)$. 
This implies that there exists a path between any two vertices among
$\{t_1, t_2, t_{ZP}\}$ in $T - \{t_P\}$, where $t_{ZP}$ is a vertex in $T$ 
such that $W(t_{ZP}) \cap ZP \neq \emptyset$. 
Since $t_P$ is adjacent to $t_1, t_2$ and $t_{ZP}$, two cycles have a common edge in $T$. This contradicts the fact that $T$ is a cactus. 
Hence, all vertices of $YP$ are in one witness set in $\calW$. 
Similar arguments show that $ZP$ is in one witness set in $\cal W$. 

Consider vertices $t_{P}, t_{YP}, t_{ZP}$ in $T$ where and 
$YP \subseteq W(t_{YP})$ and $ZP \subseteq W(t_{ZP})$. 
Clearly $t_{YP}t_P$ and $t_Pt_{ZP}$ are edges in $T$, and 
$\deg_{T}(t_P) = 2$. 
Consider a witness structure $\calW’$ obtained from $\calW$ 
by removing $W(t_P)$ and adding two new sets $\{v_1\}$ and 
$W(t_P) \setminus \{v_{1}\}$. 
Let $T’$ be the graph from $G$ by 
contracting all witness sets in $\calW’$. 
In other words, $\calW’$ is a $T’$-witness structure of $G$ and 
$F \setminus \{v_{1}v_{2}\}$ contains edges of spanning trees of 
witness sets in $\calW’$.  
As $T’$ can be obtained from $T$ by subdividing the edge 
$t_{YP}t_{P}$, by Observation~\ref{obs:cactus-prop}, 
$T’$ is also a cactus. 
This contradicts the fact that $F$ is a minimal solution. 
Hence,  then all vertices of $P$ lie in singleton witness sets. 
  
Next, we argue that $Y$ contains a big witness set in $\calW$. 
If $v_1$ has at least two neighbors in $Y$, then the above arguments
imply that all these vertices are in a single witness set of $Y$ and 
therefore $Y$ contains a big witness set. 
Otherwise, $v_1$ only has one neighbor, say $y_1$ in $Y$. 
This contradicts the maximality of $P$ in $G$. 
Hence $v_1$ is adjacent with at least two vertices in $Y$ 
which are in one big witness set. 
Similarly, we can conclude that $Z$ contains a big witness set. 
This concludes the proof of the lemma.
\end{proof}

\subsection{Additional Properties of Recoloring}
\label{subsec:exact-one-big-witness}

By the definition of compatible coloring, 
every monochromatic component contains at most one big witness set.
In Lemma~\ref{lemma:correct-recoloring2}, we argue that after recoloring, 
each monochromatic component that contains at least two vertices 
and are colored with colors in $\{1, 2, 3\}$ contains a big witness set. 
We show that \emph{all} vertices in a monochromatic component 
which does not contain a big witness set satisfy the premise of 
Lemma~\ref{lemma:recoloring-path-one-witness} or 
Lemma~\ref{lemma:recoloring-path-two-witnesses}.

\begin{lemma}\label{lemma:correct-recoloring2}
If a monochromatic component $X$ in $\calX$ which contains 
at least two vertices are colored with a color in $\{1,2,3\}$ after 
an exhaustive application of Recoloring~\ref{recolor:path-one-witness} 
and Recoloring~\ref{recolor:path-two-witnesses}, 
then $X$ contains a big witness set.
\end{lemma}
\begin{proof}
Let $T_B$ and $T_S$ be subsets of $V(T)$ which correspond 
to big witness sets and singleton witness sets, respectively.
By Observations~\ref{obs:cactus-prop}~(\ref{item:high-deg}) 
and \ref{obs:cactus-witness-cut-vertex}, any vertex in $T_S$ has 
degree at most two in $T$. 
Hence $T - T_B$ is a collection of isolated vertices and cable paths.

Let $X$ be a monochromatic component in $\calX$ which has not 
been recolored.
If for some $t$ in $T_B$, $W(t)$ is in $X$, then the lemma is true.
Assume that $X$ does not contain any vertex from $T_B$.
Let $T_X$ be the set of vertices in $T$ such that the corresponding 
witness sets contain vertices in $X$.

Any vertex in $T_S$, and hence in $T_X$, is a leaf or in the path
starting and ending at the same vertex in $T_B$ (in a pendant cycle)
or in the path connecting two different vertices in $T_B$.
Consider a vertex $t’$ in $T_X$ and $t_1, t_2$ in $T_B$. 
Let $x’ = W(t’)$ and $X’, X_1, X_2$ be the monochromatic components
containing $W(t’), W(t_1), W(t_2)$, respectively.
If $t’$ is a leaf adjacent to $t_1$ then $X’$ is a component of $G - X_1$
and hence it was recolored to $4$. 
If $t’$ is in a path starting and ending at $t_1$ then $x’$ is a in
a cable path in $G - X_1$ and hence it was recolored to $4$.
Similarly, if $t’$ is a in a path connected $t_1, t_2$ then $x’$ is 
a in a cable path in $G - (X_1 \cup X_2)$ and was recolored to $5$.  
\end{proof}

This also implies that an exhaustive application of the recoloring rules 
identifies almost all the vertices in $G$ that form singleton witness sets in $T$;
the only exceptions being vertices in some monochromatic component $X$ in 
$\calX$, which also contains a big witness set. 
In the next section, we see how to identify those singleton witness sets.

\section{Identifying Big Witness Sets} 
\label{sec:cactus-extract-big-witness}
Suppose $G$ is a 2-connected graph that is $k$-contractible 
to a cactus $T$ and $F \subseteq E(G)$ is a minimal collection of 
at most $k$ edges such that $G/F = T$. 
Let $\calW$ denote the $T$-witness structure of $G$ 
and $f$ be a compatible coloring of $G$ with respect to $\calW$ 
that has been modified by the recoloring procedure described in 
Section~\ref{sec:refinement-coloring}. 
Let $\calX$ be the set of all monochromatic components with 
respect to $f$ that are colored with $1, 2$ or $3$. 

Let $W(t_X)$ be the big witness set in $\calW$ which is 
contained in $X \in \calX$ where $|X| \geq 2$. 
The existence of $W(t_X)$ is guaranteed by Lemma~\ref{lemma:correct-recoloring2}, and its uniqueness follows the definition of compatible colorings. 
Let $\hat{X}$ denote the superset of $X$ that contains all the vertices 
in the components of $G - X$ that are either isolated vertices or 
cable paths in $G$ whose both endpoints have neighbors in $X$. 
We describe a procedure that identifies $W(t_X)$ or a set $Z_X$ 
of vertices in $G[\hat{X}]$ which is \emph{at least as good as} $W(t_X)$, 
i.e., replacing the edges in the spanning tree of $G[W(t_X)]$ with the edges
in the spanning tree of $G[Z_X]$ in $F$ leads to another set $F’$ 
such that $|F'| \leq |F|$ and $G/F'$ is a cactus. 
Then, we refine $\calX$ by deleting $X$ and adding $Z_X$ and 
$\{v\}$ for every vertex $v$ in ${\hat X} \setminus Z_X$ and 
invoke the replacement procedure with the modified $\calX$. 
This process terminates when no monochromatic component 
is replaced in the current $\calX$.

Towards this,  we first observe some properties of the 
big witness set $W(t_X)$ in $X$.
Note that since contracting it results in a cactus graph,
the vertices that are in $X \setminus W(t_X)$ 
can only be an isolated vertices or cable paths. 
Apart from this,  because of the adjacencies with
neighboring colored component,  some vertices are 
forced into this big witness set.
This motivates the following definition.
\begin{definition}[Core] \label{def:conn-core}
A \emph{core} of a graph $H$ is a subset $Z$ of $V(H)$ 
such that every  component of $H - Z$ is an isolated vertex 
or a cable path whose neighborhood is in $Z$. 
If a core $Z$ is a connected set in $H$, then 
we call it a \emph{connected core} of $H$.
\end{definition}

Notice that any superset of a connected core that 
induces a connected subgraph is also a connected core. 
The following observation is a direct consequence of the definition.

\begin{observation}\label{obs:core-to-cactus} 
For a graph $H$ and its connected core $Z$, 
let $S$ be the edges of some spanning tree of $H[Z]$.
Then, $H/S$ is a cactus.
\end{observation}

The following lemma shows that $W(t_X)$ is a connected core of $G[\hat X]$. 

\begin{lemma}\label{lemma:witness-conn-core} 
$W(t_X)$ is a connected core of $G[\hat{X}]$. 
\end{lemma}
\begin{proof} 
Since $W(t_X)$ is a witness set, 
by definition $G[W(t_X)]$ is connected. 
For the sake of contradiction, assume that $W(t_X)$ 
is not a core of $G[\hat{X}]$.
This implies that at least one component $C$ of 
$G[\hat{X}] \setminus W(t_X)$ is neither a cable path nor an isolated vertex.
Hence, $C$ contains at least $3$ vertices and there is a vertex $y$ in $C$
such that $d_{G[\hat{X}]}(y)\ge3$ and 
it is adjacent to at least two vertices in $C$.
If $y$ is in $\hat{X} \setminus X$, then by 
Lemma~\ref{lemma:recoloring-path-one-witness},  
it is in a singleton witness set. 
Otherwise, $y$ is in $X \setminus W(t_X)$ and once again it
is in a singleton witness set.
This implies that there is a vertex $t_y$ in $T$ such that 
$W(t_y) = \{y\}$ and $d_T(t_y)\ge3$.
By Observation~\ref{obs:cactus-prop}~(\ref{item:high-deg}), 
$y$ is a cut-vertex in $T$.
However, this contradicts 
Observation~\ref{obs:cactus-witness-cut-vertex} states 
that every cut-vertex in $T$ corresponds to a big witness set.
Hence,  $W(t_X)$ is a connected core of $G[\hat X]$.
\end{proof}

We point out that there may exist a proper subset of $W(t_X)$
which is a connected core of $G[\hat{X}]$. 
In other words, every vertex in $W(t_X)$ is either in a 
connected core of $G[X]$ and/or it is in $W(t_X)$ because of 
``external constraints''  as shown in Lemmas~\ref{lemma:marking1} 
and \ref{lemma:marking2}. 

\begin{lemma}\label{lemma:marking1} 
If there is a vertex $v$ in $N_G(X)$ such that $v$ is colored $5$, 
then $N_G(v) \cap X \subseteq W(t_X)$. 
\end{lemma}
\begin{proof} If $v$ is colored $5$ then by 
Lemma~\ref{lemma:recoloring-path-two-witnesses}, 
$v$ is in a cable path $P$ in $G$ between $X$ and another component 
$X’$ in $\calX$ that contains a big witness set $W(t_{X'})$ in $\cal W$. 
Further, all the vertices of $P$ are in singleton witness sets in $\calW$. 
Assume that there exists a vertex $x$ in $W(t_X) \setminus N(v)$. 
Consider a path $Q$ from $W(t_X)$ to $x$ which is entirely in $G[X]$.
Let $Q’$ be a path from $W(t_{X'})$ to an endpoint of $P$ whose 
internal vertices are in $X’$. 
See Figure~\ref{fig:marking-scheme} for an illustration.

Since $xv$ is in $G$, we know that $Q$ along with 
the edge $xv$ and paths $P$, $P’$ form a path from $W(t_X)$ 
to $W(t_{X'})$ in $G$.
This path in $G$ gives a path between $t_X$ and $t_{X’}$ in $T$
such that all the internal vertices of this path correspond to 
singleton witness sets in $\calW$.
Notice that every vertex on the path from $W(t_X)$ to $x$ 
has same color as that of $W(t_X)$.
All these vertices are in singleton witness sets.
This contradicts the fact that $f$ is a compatible coloring 
with respect to $\calW$.
Hence,  all vertices in $N(v) \cap X$ are in $W(t_X)$.
This concludes the proof of the lemma.  
\end{proof}

\begin{figure}
  \centering
  \includegraphics[scale=0.7]{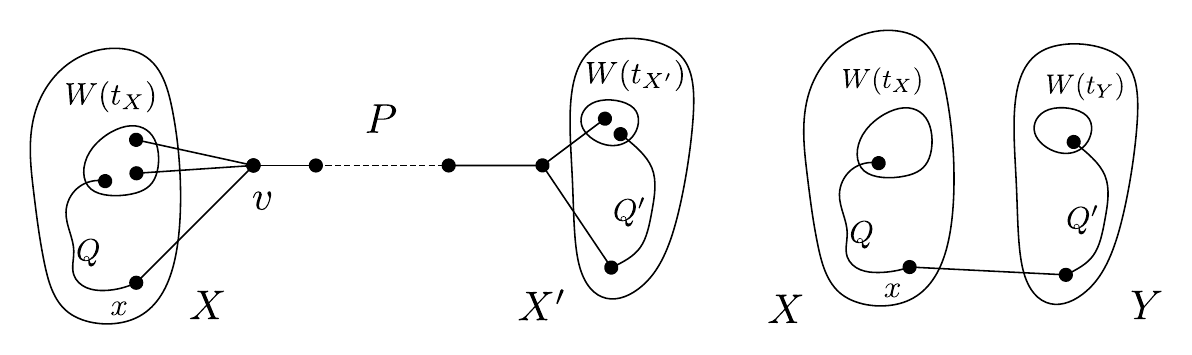}
  \caption{Illustration for Lemmas~\ref{lemma:marking1} and \ref{lemma:marking2}. 
  \label{fig:marking-scheme}}
\end{figure}

\begin{lemma}\label{lemma:marking2} 
For any monochromatic component $Y$ in $\calX$ 
which contains a big witness set $W(t_Y)$, 
$N(Y) \cap X \subseteq W(t_X)$. 
\end{lemma}
\begin{proof} 
If $E(X, Y)$ is empty, then the statement is vacuously true.
Assume that there is a vertex $x$ in 
$(N(Y) \setminus W(t_X)) \cap X$,  and let $t_X, t_Y$ be the vertices 
of $T$ corresponding to the big witness sets $W(t_X)$ and  $W(t_Y)$ respectively.
See Figure~\ref{fig:marking-scheme} for an illustration.
Since $X$ is connected, there is a path between $W(t_X)$ 
and $x$ which is entirely in $X$.
As $X$ may contain only one big witness set, 
$x$ lies in a singleton witness set in $\cal W$.
This implies that there is a path between $t_X$ 
and $t_Y$ in $T$ (via $x$) such that the neighbor of $t_X$ 
has the same color as the vertices in $W(t_X)$.
This contradicts the fact that $f$ is a compatible coloring 
with respect to $\calW$.
\end{proof}

Next, we show that a minimum connected core with specific properties
is a replacement for $W(t_X)$. Towards this, we describe 
the following marking scheme.

\begin{marking-scheme}\label{mark} 
$(i)$ If there is a vertex $y$ in $N(X)$ such that $ f(y) = 5$ then mark all the vertices in $N(y) \cap X$.
$(ii)$ For each monochromatic component $Y$ in $\calX$ such that $f(Y) \in \{1, 2, 3\}$ and $|Y| \geq 2$, mark all vertices in $N(Y) \cap X$.
\end{marking-scheme}

\begin{figure}[h!]
  \centering
  \includegraphics[scale=0.65]{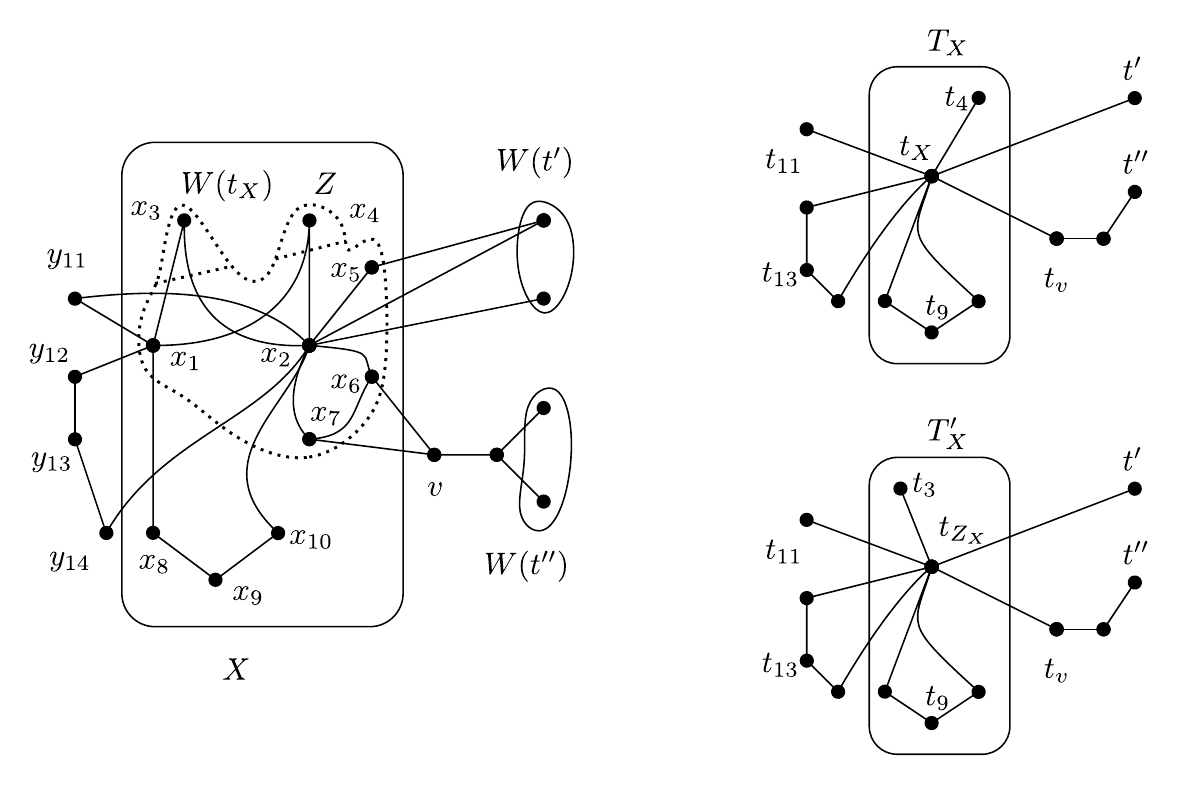}
  \caption{Replacing $W(t_X)$ by $Z$ in $\calX$ in Lemma~\ref{lemma:replacement2}. 
  \label{fig:replacement}}
\end{figure}

Let $M_X$ be the set of marked vertices in $X$. 
By Lemma~\ref{lemma:marking1} and Lemma~\ref{lemma:marking2}, 
every vertex in $M_X$ is in $W(t_X)$. 
Let $Z_X$ be \emph{any} connected core of minimum 
cardinality, which contains $M_X$. 
Let $\calX^*$ be obtained from $\calX$ by deleting $X$ 
and adding $Z_X$ and $\{v\}$ for every vertex $v$ in 
${\hat X} \setminus Z_X$. 
For each monochromatic component $X$ in $\calX^*$, 
fix a spanning tree $T_X$ of $G[X]$. 
Let $F^*$ denote $\bigcup_{X \in \calX^*} E(T_X)$.

\begin{lemma}\label{lemma:replacement2} 
$G/F^*$ is a cactus and $|F^*| \le |F|$.
\end{lemma}
\begin{proof}
Let $S_X$ be the edges of the spanning of $G[W(t_X)]$ 
that is contained in $F$ and let $S_Z$ be the edges of a 
spanning tree of $G[Z_X]$. 
Let $F'$ denote the set of edges obtained from $F$ by 
deleting $S_X$ and adding $S_Z$. 
By Lemmas~\ref{lemma:witness-conn-core},  \ref{lemma:marking1} 
and \ref{lemma:marking2}, $W(t_X)$ is a connected core of 
$G[\hat X]$ which contains $M_X$.
Since $Z_X$ is a minimum-sized connected core 
which contains $M_X$, we have $|Z_X| \le |W(t_X)|$ implying $|F’| \le |F|$. 

We now argue that $G/F’$ is a cactus. 
Let $\calW’$ be a $T’$-witness structure of $G$ where $T’ = G/F’$.
We argue that $\calW’$ can be obtained from a $T$-witness structure
$\calW$ of $G$.
We first claim that if a witness set in $\calW$ intersects 
$\hat{X}$ then it is a subset of $\hat X$.
Assume that a witness set $W(t_X)$  
intersects $\hat X$ and contains a vertex $y$ in $W \setminus \hat X$.
As $W(t_X)$ is a connected set, there is a path in $G[W(t_X)]$ 
between $y$ and $x$.
Since, $X$ is a separator between $\hat X \setminus X$ 
and $V(G) \setminus \hat X$, this path intersects $X$.
This implies that there is a witness set with a vertex in $X$ 
and a vertex outside $X$.
Since every witness set is monochromatic, no such witness set exists.
Hence, no witness set contains $\hat X$ and 
vertex outside $\hat X$.
This implies $\calW’$ can be obtained from $\calW$ by 
removing all witness sets in ${\hat X}$ and adding $Z_X$ 
and a singleton witness set for each vertex in $\hat X \setminus Z$. 
Formally, $\calW’ = (\calW \setminus \calW_{\hat X}) \cup \calW’_1$ 
where $ \calW’_1 = \{Z_X\} \cup \{\{v\}|\ v \in \hat{X} \setminus Z_X\}$
and $\calW_{\hat X}$ is the set of all witness sets contained in $\hat X$. 
Refer to Figure~\ref{fig:replacement} for an illustration.

Let $T_X$ be the subgraph induced on the vertices of $T$ 
whose corresponding witness sets are in $\hat X$. 
That is, let $V_1 = \{t \mid \ t \in V(T) \text{ and } W(t_X) \subseteq \hat X\}$
and $T_X = T[V_1]$.
We similarly define $T’_X$. Since $\calW$ and $\calW’$ 
are $T$-witness structure and $T’$-witness structure, respectively, 
of $G$, the graphs $T - V(T_X)$ and $T’ -V(T’_X)$ are isomorphic 
to each other.

Recall that $W(t_X)$ is the big witness set in $X$.
Let $t_Z$ be the vertex in $T’$ such that $W’(t_Z) = Z_X$.
We now argue that $T - (V(T_X) \setminus \{t_X\})$ and 
$T’ - (V(T’_X) \setminus \{t_Z\})$ are isomorphic.
It suffices to prove that the neighbors of $t_X$ in $T - V(T_X)$ 
are identical to that of $t’_x$ in $T’ - V(T’_X)$, i.e., 
$N_{T’}(t_Z) \setminus V(T’_X) = N_T(t_X) \setminus V(T_X)$.
Consider a vertex $x$ in $\hat X$ which has a neighbor $y$ in
$V(G) \setminus \hat X$.
Let $t_y$ be a vertex in $T$ such that $y$ is in $W(t_y)$. 
There are three possibilities for $t_y$ in $T$. $(1)$ $t_y$ is 
in a path between $t_X$ and some other vertex in $T$ 
which corresponds to a big witness set $(2)$ $t_y$ is in 
a pendant cycle in which $t_X$ is the unique (cut) vertex 
which corresponds to a big witness set $(3)$ $t_y$ is a 
leaf adjacent to $t_X$.
In the first case, $y$ is colored $5$, and hence $x$ is in $M_X$.
In the second and third cases, $y$ is an isolated vertex or in a 
cable path in $G - W(t_X)$ and hence in $G - X$.
This implies that $y$ is in $\hat X$ contradicting that $y$
is in $V(G) \setminus \hat X$. 
Hence $M_X$ contains every vertex in $X$ that has a neighbor 
in $V(G) \setminus \hat X$. 
This implies both $W(t_X)$ and $Z$ contain every vertex in 
$X$ that has a neighbor in $V(G) \setminus \hat X$, and 
therefore $N_{T’}(t_Z) \setminus V(T’_X) = N_T(t) \setminus V(T_X)$.
Hence $T - (V(T_X) \setminus \{t_X\})$ and 
$T’ - (V(T’_X) \setminus \{t_Z\})$ are isomorphic.

By Observation~\ref{obs:core-to-cactus}, 
both $G[\hat X]/S_t$ and $G[\hat X]/S_Z$ are cacti. 
Once again, since $M_x$ (hence $W(t_X)$ and $Z$) contains all 
vertices in $\hat X$ which has neighbors outside, $t_X$ and $t_Z$ 
are the only vertices in $T$ and $T’$, respectively, 
which have neighbors outside $T_X$ and $T’_X$, respectively. 
$T’ - (V(T’_X)\setminus \{t_Z\})$ is cactus as it is isomorphic to 
$T - (V(T_X) \setminus \{t_X\})$ which is a cactus. 
Since $T’_X$ is also a cactus and $t_Z$ is the only vertex which 
has neighbors outside $T’_X$, $T’$ is a cactus.
This concludes the proof that $G/F’$ is a cactus. 
\end{proof}

\subsection{Finding Minimum Connected Cores}
\label{subsec:core}

Recall that a connected core of a graph $G$ is a connected set 
of vertices of $G$ such that every component of $G - Z$ is an 
isolated vertex or a cable path whose both endpoints have neighbors in $Z$. 
In this subsection, we present a simple branching algorithm to compute 
a minimum connected core containing a given subset of vertices. 
We use an algorithm known for the \textsc{Steiner Tree} problem as 
a subroutine. In the \textsc{Steiner Tree} problem, 
we are given a graph, a set of vertices called \emph{terminals}, 
and a positive integer $\ell$, and the goal is to determine whether 
there is a tree with at most $\ell$ edges that contain all the terminals. 
This problem admits an algorithm with $\mathcal{O}^*(2^{\ell})$ 
running time \cite{nederlof2013fast} where $\ell$ is the number of terminals.
 
\begin{lemma}
\label{thm:conn-core-algo}
There is an algorithm that given a connected graph $G$ on $n$ vertices, 
a subset $X$ of its vertices and an integer $k$, 
either computes a minimum connected core of $G$ 
which contains $X$ and is of size at most $k$ or 
correctly concludes that no such connected core exists 
in $6^k \cdot n^{\calO(1)}$ time.
\end{lemma}
\begin{proof}
We first construct a set of candidate cores of $G$ via a 
branching algorithm.  
At each leaf of the branching tree, we extend the core constructed 
to a connected set using the \textsc{Steiner Tree} algorithm.

Let $Z$ denote a partial solution to the instance of finding a connected core.
Initialize $Z$ to $X$ and decrease $k$ by $|X|$.
The following branching rule is derived from the observation 
that if $(u,v,w)$ is a path outside the connected core of $G$ 
then the degree of $v$ is two in $G$.
\begin{branching rule}
\label{br1}
If there is a path $(u,v,w)$ in $G-Z$ such that $|N_G(v)| \geq 3$, 
then branch into three cases where $u$ or $v$ or $w$ is added 
to $Z$ and decrease $k$ by one in each branch.
\end{branching rule}
Observe that when this rule is no longer applicable, all vertices of 
$G-Z$ have the degree at most two.
Hence the components of $G-Z$ are cable paths in $G$ or isolated vertices.
Next, we have the following reduction rule from the observation: if $xy$ is an isolated edge obtained by removing 
the minimum connected core of $G$, then $x$ and $y$ have 
degree two or more in $G$.
\begin{reduction rule}
\label{rr1}
If there is an edge $uv$ in $G-Z$ such that $u$ is the unique 
neighbor of $v$ then add $u$ into $Z$ and reduce $k$ by one. 
\end{reduction rule}
Since the only neighbor of $v$ is $u$, the edge $uv$ cannot be
in a cable path in $G$ whose both endpoints have neighbors in $Z$.
Now, if there exists an optimal solution $Z^*$ that does not contain $u$, 
then $v \in Z^*$ and $Z’ = (Z^* \setminus \{v\}) \cup \{u\}$ is also a 
the connected core of $G$.
This justifies the correctness of the rule.

We apply the above rules exhaustively and consider the search tree 
constructed.
Note that each node of the search tree is labeled with either a triple 
$(u,v,w)$ showing that the Branching Rule~\ref{br1} was applied, 
or a pair $(x,y)$ showing that Reduction Rule~\ref{rr1} was applied at this node.
If at any node in the search tree, $k$ is $0$ and the set $Z$ is not a connected core of $G$, we abort the computation at that node.
If all the leaves of the search tree are aborted, then we output 
\noalgo\ as a solution to this instance.

Next, we claim that if none of the rules are applicable at a leaf of 
the search tree, then the corresponding $Z$ is a core of $G$.
Assume that $Z$ is not a core of $G$.
Then there is a component $C$ of $G-Z$ that is neither an isolated
vertex nor a cable path in $G$ whose both endpoints have neighbors in $Z$. 
Hence such a $C$ has at least two vertices. 
Furthermore as Branching Rule~\ref{br1} is not applicable, all vertices in $G-Z$
have a maximum degree of $2$. 
First, consider the case when $C$ is a cycle in $G-Z$. 
As $G$ is connected, $C$ has a vertex $v$ with a neighbor in $Z$. 
Let $u$ and $w$ be the neighbours of $v$ in $C$. 
Then, $(u,v,w)$ is a path in $G-Z$ with $|N_G(v)| \geq 3$ leading to a contradiction. 
Next, consider the case when $C$ is a path in $G-Z$ with endpoints $u$ and $v$. 
If there is an internal vertex on this path that has a neighbor in $Z$, then, as before, we get a contradiction. 
Hence, $C$ is a cable path in $G$, with endpoints $u$ and $v$.
As $Z$ is not a core of the connected graph $G$, one of $u$ or $v$ 
has no neighbour in $Z$, i.e. it is a vertex of degree $1$ in $G$. 
This implies that Reduction rule \ref{rr1} is applicable, leading to a contradiction. 
Hence $Z$ is a core of $G$.

Consider a leaf of the search tree and the corresponding core $Z$. 
As $Z$ may not be connected in $G$, we may have to add additional 
vertices to $Z$ to ensure being connected. 
Observe that this can be achieved by computing a minimum 
\textsc{Steiner Tree} with $Z$ being the terminal set in $G$ in $\OO(2^k)$
time \cite{nederlof2013fast}. 
Let $Z’ \supseteq Z$ be the vertices of this Steiner tree obtained 
by this algorithm. 
Observe that $Z’$ is a connected core of $G$, as $G-Z$ is a collection 
of isolated vertices and cable paths of $G$. 
Let $\bar{Z}$ be the minimum cardinality connected core over all 
the leaves of the search tree. If $|\bar{Z}| \leq k$, we output $\bar{Z}$ 
and otherwise, we output \noalgo\ as the solution to the instance.

Let us now argue the correctness of this algorithm. 
Assume $Z^*$ is an optimal solution of size at most $k$. 
We claim that the above algorithm finds a connected core $\bar{Z}$ 
such that $|\bar{Z}| \le |Z^*|$. 
To argue this, we associate a path in the search tree of branching
algorithm to the set $Z^*$. 
Now consider an internal node in the search tree labeled with $(a, b, c)$. 
Since Branching Rule~\ref{br1} is applied at this node, $(a, b, c)$ is a path 
in $G-Z$ and $|N_G(b)| \ge 3$.
As $Z^*$ is a core of $G$, at least one of $a, b, c$ must be present. 
Similarly, for any node labeled with a pair $(x,y)$, one of these vertices, 
say $y$, is of degree $1$ in $G$, and hence $Z^*$ must contain one of them.
Recall that, by previous arguments, we may assume that $x \in Z^*$.
Hence, we start from the root of the search tree and navigate to a 
leaf along the edges corresponding to choices consistent with $Z^*$.
If multiple choices are consistent with $Z^*$,  we arbitrarily pick 
one of them and proceed.
Consider the set $\tilde{Z}$ obtained at the leaf via this navigation 
consistent with $Z^*$ from the root of the search tree.
Clearly, $\tilde{Z} \subseteq Z^*$ and $\tilde{Z}$ is a core 
(not necessarily connected) of $G$.
Let $T$ be an optimal solution for an instance of $(H, \tilde{Z})$ 
of \textsc{Steiner Tree} as defined above.
Since, $Z^*$ is a connected core of $G$ and $\tilde{Z} \subseteq Z^*$ 
we know that $Z^* \setminus \tilde{Z}$ is a solution to this 
\textsc{Steiner Tree} instance.
By the optimality of $T$, $|T| \le |Z^* \setminus \tilde{Z}|$ and 
hence $\bar{Z} = \tilde{Z} \cup T$ is the desired solution.

Let us now analyze the running time of this algorithm.
At each application of the Branching Rule~\ref{br1}, 
we have a three-way branch, and the measure drops by $1$ resulting
in the branching vector $(1,1,1)$.
This leads to the recurrence $T(k) \le 3T(k - 1)$ with solution 
$T(k)=3^k\cdot n^{\calO(1)}$.
Then, at each leaf of the search tree, we run the algorithm for 
finding a minimum Steiner tree, which runs in $2^k \cdot n^{\calO(1)}$ time.
If the Steiner tree obtained is of size strictly more than 
$k$, then we discard this node. 
Therefore, the overall running time is $6^k \cdot n^{\calO(1)}$. 
\end{proof}

\section{The Overall Algorithm}
\label{sec:cactus-algo-together}

In this section, we present a single exponential time \FPT\ 
algorithm to solve \textsc{Cactus Contraction}.
As mentioned before, we first present a randomized algorithm
when the input graph is $2$-connected (Lemma~\ref{lemma:cactus-con-2-conn}). 
Then, we argue that a connected graph is $k$-contractible 
to a cactus if and only if each of its $2$-connected components 
is contractible to a cactus using at most $k$ edge contractions 
in total (Lemma~\ref{lemma:2-conn-cactus}). 
Using this result and the algorithm for 2-connected graphs, 
we present a randomized algorithm on general graphs 
(Theorem~\ref{thm:cactus-con-rand}). 
Finally, we describe how to derandomize this algorithm 
using $(n,k)$-universal sets (Theorem~\ref{thm:cactus-con-det}).

\subsection{A Randomized \FPT\ Algorithm}

Consider a $2$-connected graph $G$ and an integer $k$. 
The algorithm consists of the following four steps and returns
whether or not $(G,k)$ is a \yesinst\ instance  of \textsc{Cactus Contraction}.

\begin{enumerate}
\item Construct a coloring $f: V(G) \rightarrow \{1, 2, 3\}$ of $V(G)$
by assigning one of the colors to each vertex independently 
and uniformly at random.

\item Expand the range of coloring $f$ from $\{1, 2, 3\}$ to $\{1, 2, 3, 4, 5\}$
using Recoloring~\ref{recolor:path-one-witness} and 
Recoloring~\ref{recolor:path-two-witnesses}. 
Recall that the recoloring procedure exhaustively applies 
Recoloring~\ref{recolor:path-one-witness} before starting 
Recoloring~\ref{recolor:path-two-witnesses}. 
Also, once it recolors a vertex to $4$, it does not change its color. 

\item 
For a monochromatic component $X$ in $\calX$ with 
$f(X) \in \{1, 2, 3\}$, apply Marking Scheme~\ref{mark}. 
Let $M_X$ be the set of marked vertices in $X$. 
Compute a connected core of $G[\hat X]$, say $Z_X$, 
of minimum cardinality that contains $M_X$. 
Here, $\hat{X}$ is the superset of $X$ that contains all 
vertices in the connected components of $G - X$ that are
isolated vertices or cable paths in $G$ whose both endpoints
have neighbors in $X$. Modify the partition $\calX$ by removing $X$
and adding $Z_X$ and $\{v\}$ for every vertex $v$ in 
${\hat X} \setminus Z_X$. 
Repeat until no monochromatic component is replaced in $\calX$.

\item
For each monochromatic component $X$ in $\calX$, 
fix a spanning tree $T_X$ of $G[X]$. 
Let $F$ denote $\bigcup_{X \in \calX} E(T_X)$. 
If $|F| \le k$ and $G/F$ is a cactus, then the return 
\yesalgo; otherwise, return \noalgo.
\end{enumerate}

\begin{lemma} \label{lemma:cactus-con-2-conn} 
Given an instance $(G, k)$ of \textsc{Cactus Contraction} 
where $G$ is a $2$-connected graph on $n$ vertices, 
there is a one-sided error Monte Carlo algorithm with 
false negatives which determines whether $(G, k)$ is a 
\yesinst\ instance in  $2^{\calO(k)} \cdot n^{\mathcal{O}(1)}$ time. 
Further, the algorithm returns the correct answer with constant probability.
\end{lemma}
\begin{proof}
Suppose $(G, k)$ is a \yesinst\ instance and $G$ is 
$k$-contractible to a cactus $T$. 
Fix a $T$-witness structure $\calW$ of $G$ and 
let $F$ be the set of edges such that $G/F=T$. 
We first show that if the coloring $f$ chosen in Step~1 
is a coloring compatible with $\calW$, 
then the algorithm returns the correct answer. 

Lemmas~\ref{lemma:recoloring-path-one-witness} and 
\ref{lemma:recoloring-path-two-witnesses} imply the correctness of 
Step~2 that uses Recoloring~\ref{recolor:path-one-witness} 
and Recoloring~\ref{recolor:path-two-witnesses}. 
Let $\calX$ be the collection of monochromatic components 
of $f$ after exhaustive application of these two recoloring rules. 
By Lemma~\ref{lemma:correct-recoloring2}, if a monochromatic 
component $X$ in $\calX$ contains at least two vertices and 
$f(X) \in \{1,2,3\}$, then $X$ contains a big witness set. 
Further, as \emph{all} vertices in a monochromatic component 
which does not contain a big witness set satisfy the premise of 
Lemma~\ref{lemma:recoloring-path-one-witness} or 
Lemma~\ref{lemma:recoloring-path-two-witnesses},  
the recoloring procedure never recolors a vertex that belongs 
to a big witness set in $\calW$. 
By Lemma~\ref{lemma:recoloring-path-one-witness},  every vertex 
recolored to $4$ is a singleton witness set in $\calW$. 
By Lemma~\ref{lemma:recoloring-path-two-witnesses}, 
every vertex recolored to $5$ is a singleton witness set 
in $\calW$. 
This proves the correctness of Step~2. 

Therefore, at the beginning of Step 3, 
each monochromatic component $X$ which contains at 
least two vertices also contain a big witness set $W(t_X)$ 
(Lemma~\ref{lemma:correct-recoloring2}). 
Using Lemma~\ref{thm:conn-core-algo},  we identify a set $Z_X$ 
of vertices which is \emph{at least as good as} $W(t_X)$ by 
Lemma \ref{lemma:replacement2}. 
That is, replacing the edges in the spanning tree of $G[W(t_X)]$ 
with the edges in the spanning tree of $G[Z_X]$ in $F$ leads to 
another set $F’$ such that $|F'| \leq |F|$ and $G/F'$ is a cactus. 
Then, Step 3 refines $\calX$ by deleting $X$ and adding $Z_X$ 
and $\{v\}$ for every vertex $v$ in ${\hat X} \setminus Z_X$ and 
proceeds with invoking the replacement procedure with the 
modified $\calX$. 
Step 3 terminates when no monochromatic component is 
replaced in the current $\calX$.

Hence, if $f$ is compatible with $\calW$,  
the algorithm returns the correct answer. 
By Lemma~\ref{lemma:random-coloring-compatible}, 
a random $3$-coloring $f$ of $G$ is compatible with 
$\calW$ with probability at least $1\large/3^{6k}$. 
So we run the algorithm with input $(G,k)$ $3^{6k}$ 
times and declare that $(G, k)$ is a \yesinst\ instance 
if and only if at least one of these runs returns \yesalgo. 
The probability that none of these colorings 
(each of which is generated uniformly at random) 
is compatible with $\calW$ is at most 
$(1 - \frac{1}{3^{6k}})^{3^{6k}} < 1/e$. 
That is, one of the randomly generated colorings is 
compatible with $\calW$ with probability at least $1-1/e$. 
As the algorithm returns \yesalgo\ only if it has found a set 
of edges with the desired properties,  the algorithm returns 
\yesalgo\ only if the given instance is a \yesinst\ instance of 
\textsc{Cactus Contraction}. 
Hence the algorithm returns the correct answer with the probability
at least $1-1/e$. 
Finally, observe that Steps 1, 2, and 4 take polynomial time while 
Step 3 takes $6^k \cdot n^{\calO(1)}$ time. 
Hence the total running time of the algorithm is 
$c^k \cdot n^{\mathcal{O}(1)}$ for a fixed constant $c$.
\end{proof}
 
\begin{lemma} \label{lemma:2-conn-cactus}
A connected graph is $k$-contractible to a cactus 
if and only if each of its $2$-connected components 
is contractible to a cactus using at most $k$ edge contractions.
\end{lemma}
\begin{proof} 
We prove the claim by induction on the number of vertices
in the graph.  
The claim trivially holds for a graph on a single vertex. 
Assume that the claim holds for graphs with less than 
$n$ vertices. 
Consider a connected graph $G$ on $n$ vertices.
Suppose $G$ is $k$-contractible to a cactus.
Then, there is a set $F \subseteq E(G)$ of size at most $k$ 
such that $T=G/F$ is a cactus.
Let $\mathcal{W}$ be the corresponding $T$-witness structure of $G$.
Let $v$ be a cut-vertex in $G$, and $C$ be a connected component of $G-\{v\}$.
Let $G_1$ denote the subgraph of $G$ induced on $V(C) \cup \{v\}$
and $G_2$ denote the subgraph of $G$ induced on $V(G)\setminus V(C)$.
Then, $G_1$ and $G_2$ are connected graphs satisfying $V(G_1) \cap V(G_2) = \{v\}$.
Further, the sets $E(G_1)$ and $E(G_2)$ partition $E(G)$.

We claim that $G_1/(F \cap E(G_1))$ and $G_2/(F \cap E(G_2))$ are both cacti.
Consider the vertex $t_0 \in V(T)$ such that $v \in W(t_0)$. 
As the deletion of a vertex in $G_2 -\{v\}$ cannot disconnect $G_1$, 
every set in 
$\mathcal{W}_1 = \{W(t) \setminus V(G_2) \mid \ t \neq t_0, W(t) \in \mathcal{W} \} 
\cup \{ W(t_0) \setminus (V(G_2) \setminus \{v\})\}$ 
induces a connected subgraph of $G$.
Then, $F \cap E(G_1)$ is the associated set of solution edges
and $G_1/(F \cap E(G_1))$ is the subgraph of $G/F$ induced 
on $\{t \in V(T) \mid W(t) \cap V(G_1) \neq \emptyset \}$.
Since an induced subgraph of a cactus is also a cactus, 
$G_1/(F \cap E(G_1))$ is a cactus.
A similar argument holds for $G_2/(F \cap E(G_2))$.
As $E(G_1)$ and $E(G_2)$ form a partition of $E(G)$, 
$|F \cap E(G_1)| + |F \cap E(G_2)| \le k$.
By the induction hypothesis, the required claim holds for 
$G_1$ and $G_2$, and the result follows.

Conversely, let $G_1, G_2, \dots G_l$ be the $2$-connected
components of $G$ and let $F_i \subseteq E(G_i)$ be a set 
of edges such that $G_i/F_i$ is a cactus and $\sum_{i \in [l]} |F_i| \le k$.
Let $\mathcal{W}_i$ be the $G_i/F_i$-witness structure of $G_i$.
Define $\mathcal{W} = \bigcup_{i \in [l]} \mathcal{W}_i$. 
Now, we make $\mathcal{W}$ into a partition of $V(G)$ 
using the following step. 
If a vertex $v$ is in $W(t_1)$ and in $W(t_2)$ then 
add $W(t_{12}) = W_{1} \cup W_{2}$ to $\mathcal{W}$ 
and delete both $W(t_1)$ and $W(t_2)$.
Then, $F = \bigcup_{i \in [l]}F_i$ contains the edges 
of a spanning tree of every witness set in $\mathcal{W}$
and $|F| \le k$.
It remains to argue that $G/F$ is a cactus.
If $G/F$ is not a cactus, then there are two cycles 
$C_1, C_2$, which share at least two vertices.
As any cycle can have vertices from only a single 
$2$-connected component of a graph, $C_1, C_2$ are 
both in some $2$-connected component of $G$ 
leading to a contradiction. 
\end{proof}
 
Now, we use Lemmas \ref{lemma:cactus-con-2-conn} 
and \ref{lemma:2-conn-cactus} to solve 
\textsc{Cactus Contraction} on general graphs. 

\begin{theorem}
\label{thm:cactus-con-rand} 
Given an instance $(G, k)$ of \textsc{Cactus Contraction} 
where $G$ is a connected graph on $n$ vertices, 
there is a one-sided error Monte Carlo algorithm with 
false negatives which determines whether $(G, k)$ is 
a \yesinst\ instance in  $2^{\calO(k)} n^{\mathcal{O}(1)}$ time. 
Further, the algorithm returns the correct answer with constant probability.
\end{theorem}
\begin{proof}
Let $G_1, G_2, \dots, G_q$ be the $2$-connected components 
of $G$ that are not cactus. 
If $q \le 1$ then we use Lemma~\ref{lemma:cactus-con-2-conn} 
to determine if $(G, k)$ is a \yesinst\ instance. 
If $q \ge k + 1$ then we return \noalgo\ as at least one edge 
needs to be contracted in each of these $2$-connected components. 

We now consider the case when $2 \le q \le k$. 
For each $G_i$ and each possible value $k_j$ between $1$ 
and $k$, we use Lemma \ref{lemma:cactus-con-2-conn} 
$3\log k$ times to determine if $(G_i, k_j)$ is a \yesinst\ instance. 
As there are at most $k^2$ pairs $(G_i,k_j)$, 
Lemma~\ref{lemma:cactus-con-2-conn} is invoked at most 
$3k^2\log k$ times. 
If the answer is \noalgo\ in all runs for all the values of $k_j$ 
for some $G_i$ then we return \noalgo. 
Otherwise, let $k_i’$ be the smallest value for which the algorithm
returns a solution for $G_i$. 
Since the algorithm in Lemma~\ref{lemma:cactus-con-2-conn}  
returns no false positive, $G_i$ is $k’_i$-contractible to a cactus.

On the other hand if $(G_i, k_i)$ is a \yesinst\ instance of 
\textsc{Cactus Contraction} then the probability that no 
run will output the right answer is at most $(\frac{1}{e})^{3\log k} = \frac{1}{k^3}$.
Since there are at most $k^2$ pairs $(G_i, k_j)$, 
and by the union bound on probabilities, 
the probability that there is a pair $(G_i , k_j)$ 
for which the algorithm returns the wrong answer is 
upper bounded by $k^2 \cdot \frac{1}{k^3} \ge \frac{1}{k}$. 
If such a failure does not occur, then for every $i$ we have 
that $k_i’$ is the smallest value of $k_j$ such that $G_i$ is $k_i$-contractible to a cactus.
Finally, the algorithm answers \yesalgo\ only if 
$\sum_{i = 1}^{q}k’_i \le k$,  and answers \noalgo\ otherwise.
The correctness of this step follows from Lemma~\ref{lemma:2-conn-cactus}. 
Consequently, the algorithm cannot give false positives, and it may 
give false negatives with probability at most $1/k \le 1/q \le 1/2$, 
where the two inequalities follow from the assumption that $2 \le q \le k$.
\end{proof}

\subsection{A Deterministic \FPT\ Algorithm}
\label{subsec:derand}
In this section, we show how to derandomize the algorithm 
described in Theorem~\ref{thm:cactus-con-rand} using a family 
of coloring functions that are derived from a universal set.

\begin{definition}[Universal Set]
A $(n,k)$-universal set is a family $\calH$ of subsets of $[n]$ 
such that for any $S \subseteq [n]$ of size at most $k$, 
$\{ S \cap H ~|~ H \in \calH \}$ contains all subsets of $S$.    
\end{definition}

Given integers $n, k$, one can construct a $(n, k)$-universal set 
using the following result.
  
\begin{proposition}[\cite{Naor-splitters}] 
\label{prop:universal-set}
For any $n,k \geq 1$, we can construct a $(n,k)$-universal set
of size $2^k k^{\calO(\log k)} \log n$ in time $2^k k^{\calO(\log k)} n \log n$.
\end{proposition}

We use this $(n,k)$-universal set to construct a $3$-coloring family of $V(G)$.

\begin{lemma}\label{lemma:coloring-family}
Consider a graph $G$ and a subset $S$ of $V(G)$ of size $6k$. 
There is a family of $3$-coloring functions $\calF$,  such that for 
a given $3$-coloring $f$ of $V(G)$, there exists coloring $\psi$ 
in $\calF$ that agrees with $f$ on $S$. 
This family has size $4^{6k} k^{\calO(\log k)} \log^2 n$ 
and can be constructed in time $4^{6k} k^{\calO(\log k)} n \log^2 n$.
\end{lemma}
\begin{proof} 
Let $\calH$ be a $(n,6k)$-universal set that is constructed by 
Proposition~\ref{prop:universal-set}.
We define a family of partitions of $V(G)$ as follows.
$$\calF’ = \{ (A,B,C) ~|~ A \in \calH,\, B = Y \setminus 
A \textit{ where } Y \in \calH,\, C = V(G) \setminus Y\}$$
Observe that $\calF’$ can be constructed by considering 
each pair of sets in $\calH$. Let $f$ be a $3$-coloring that 
partitions $S$ into $3$ parts, say $S_1,S_2,S_3$. 
We claim that there is a triple $(A,B,C) \in \calF’$ such that 
$S \cap A = S_1$, $S \cap B = S_2$ and $S \cap C = S_3$.
Since $\cal H$ is a $(n,6k)$-universal-set, there is some set 
$Y \in \calH$ such that $S \cap Y = S_1 \cup S_2$, and there
is some $A \in \calH$ such that $A \cap S = S_1$.
Hence, $S \cap (Y-A) = (S \cap Y) \setminus (S \cap A) = S_2$. 
We can easily convert the family $\calF’$ into a family of 
coloring functions, where for each $(A,B,C) \in \calF’$ maps all 
vertices in $A,B,C$ to $1,2,3$ respectively.
Hence if $f$ partitions $S$ into $S_1, S_2, S_3$, then there is 
a function $\psi \in \calF$, which also partitions $S$ into $S_1, S_2, S_3$.
Since the family $\calH$ has size $2^{6k} k^{\calO(\log k)} \log n$, 
the size of $\calF$ is at most $4^{6k} k^{\calO(\log k)} \log^2 n$ 
and it can be constructed in time $4^{6k} k^{\calO(\log k)} n \log^2 n$. 
\end{proof}

In the algorithm mentioned in Theorem~\ref{thm:cactus-con-rand}, 
instead of repeatedly generating random coloring, 
we use the coloring family mentioned in Lemma~\ref{lemma:coloring-family}
to get the following result. 

\begin{theorem} \label{thm:cactus-con-det} 
Given an instance $(G, k)$ of \textsc{Cactus Contraction} where
$G$ is a connected graph on $n$ vertices, there is a deterministic
algorithm which determines whether $(G, k)$ is a \yesinst\ instance 
in $2^{\calO(k)} \cdot n^{\mathcal{O}(1)}$ time.
\end{theorem}

\section{Conclusion}
In this article, we presented a single exponential time \FPT\ algorithm
for \textsc{Cactus Contraction}. 
Given an instance $(G, k)$, the algorithm runs in 
$2^{\calO(k)}\cdot |V(G)|^{\mathcal{O}(1)}$ time where $c$ is a 
fixed constant and correctly determines whether it is a \yesinst\ instance. 
A natural future direction is to optimize $c$ and to study other 
contraction problems that admit single exponential time algorithms.

It will be interesting to see if we can obtain a single exponential time
algorithm for \textsc{(Treewidth $\le \eta$)-Edge Contraction}
for some fixed constant $\eta$.
In this problem, the input is graph $G$ and an integer $k$ and the 
objective is to check whether one can contract at most $k$ edges 
in $G$ to obtain a graph whose treewidth is at most $\eta$.
Our original motivation to study \textsc{Cactus Contraction} 
came from understanding the case when $\eta = 2$.
Towards this, the next natural steps are to obtain such improved
algorithm for \textsc{Outer-planar Contraction} and \textsc{Series-Parallel Contraction}.

\section*{Acknowledgment: } We would like to thank Prof.  Saket Saurabh
for invaluable advice and several helpful suggestions.


%\section*{References}
%\bibliographystyle{apa}
\bibliographystyle{elsarticle-num}
 
\bibliography{ref}


%\begin{thebibliography}{ref.bib}


%\bibitem{}

%\end{thebibliography}
\end{document}